\def\enableLlncsCode{}

\def\nicerThanLlncs{}

\documentclass[\ifdefined\nicerThanLlncs envcountsect, envcountsame\fi]{llncs}

\input{include}

\title{Security Limitations of Classical-Client Delegated Quantum Computing}
\date{}

\author{\,}
\institute{\,}

\author{Christian Badertscher \inst{1} , Alexandru Cojocaru \inst{2}, L\'{e}o Colisson\inst{3} , Elham Kashefi \inst{2,3}, Dominik Leichtle\inst{3}, Atul Mantri \inst{4}, Petros Wallden\inst{2} }

\institute{IOHK, Zurich, Switzerland \\
\href{mailto:christian.badertscher@iohk.io}{christian.badertscher@iohk.io}\\
\and
School of Informatics, University of Edinburgh, \\
10 Crichton Street, Edinburgh EH8 9AB, UK \\ \href{mailto:ekashefi@inf.ed.ac.uk}{ekashefi@inf.ed.ac.uk}, \href{mailto:a.d.cojocaru@sms.ed.ac.uk}{a.d.cojocaru@sms.ed.ac.uk}, \href{mailto:petros.wallden@ed.ac.uk}{petros.wallden@ed.ac.uk}
\and
Laboratoire d'Informatique de Paris 6 (LIP6), Sorbonne Universit\'{e}, \\
4 Place Jussieu, 75252 Paris CEDEX 05, France \\ \href{mailto:leo.colisson@lip6.fr}{leo.colisson@lip6.fr}, \href{dominik.leichtle@lip6.fr}{dominik.leichtle@lip6.fr}
\and
Joint Center for Quantum Information and Computer Science (QuICS),\\ University of Maryland, College Park, USA \\ \href{mailto:amantri@umd.edu}{amantri@umd.edu}
}

\begin{document}

{\def\addcontentsline#1#2#3{}\maketitle}

\begin{abstract}
Secure delegated quantum computing is a two-party cryptographic primitive, where a computationally weak client wishes to delegate an arbitrary quantum computation to an untrusted quantum server in a privacy-preserving manner. Communication via quantum channels is typically assumed such that the client can establish the necessary correlations with the server to securely perform the given task. This has the downside that all these protocols cannot be put to work for the average user unless a reliable quantum network is deployed.

Therefore the question becomes relevant whether it is possible to rely solely on classical channels between client and server and yet benefit from its quantum capabilities while retaining privacy. Classical-client remote state preparation ($\sf{RSP}_{CC}$) is one of the promising candidates to achieve this because it enables a client, using only classical communication resources, to remotely prepare a quantum state. However, the privacy loss incurred by employing $\sf{RSP}_{CC}$ as sub-module to avoid quantum channels is unclear.

In this work, we investigate this question using the Constructive Cryptography framework by Maurer and Renner~\cite{maurer2011abstract}.  We first identify the goal of $\sf{RSP}_{CC}$ as the construction of ideal \RSP resources from classical channels and then reveal the security limitations of using $\sf{RSP}_{CC}$ in general and in specific contexts:

\begin{enumerate}
    \item We uncover a fundamental relationship between constructing ideal \RSP resources (from classical channels) and the task of cloning quantum states with auxiliary information. Any classically constructed ideal \RSP resource must leak to the server the full classical description (possibly in an encoded form) of the generated quantum state, even if we target computational security only. As a consequence, we find that the realization of common \RSP resources, without weakening their guarantees drastically, is impossible due to the no-cloning theorem.
    \item The above result does not rule out that a specific $\sf{RSP}_{CC}$ protocol can replace the quantum channel at least in some contexts, such as the Universal Blind Quantum Computing ($\sf{UBQC}$) protocol of Broadbent et al.~\cite{broadbent2009universal}. However, we show that the resulting $\sf{UBQC}$ protocol cannot maintain its proven composable security as soon as $\sf{RSP}_{CC}$ is used as a subroutine.
    \item  We show that replacing the quantum channel of the above $\sf{UBQC}$ protocol by the $\sf{RSP}_{CC}$ protocol QFactory of Cojocaru et al.~\cite{cojocaru2019qfactory}, preserves the weaker, game-based, security of $\sf{UBQC}$.
\end{enumerate}

\end{abstract}

\ifdefined\nicerThanLlncs%
  \newpage%
  \tableofcontents%
  \newpage%
\fi

\section{Introduction}
The expected rapid advances in quantum technologies in the decades to come are likely to further disrupt the field of computing. To fully realize the technological potential, remote access, and manipulation of data must offer strong privacy and integrity guarantees and currently available quantum cloud platform designs have still a lot of room for improvement.

There is a large body of research that exploits the client-server setting defined in \cite{childs2005secure} to offer different functionalities, including secure delegated quantum computation~\cite{broadbent2009universal,morimae2012blind,dunjko2014composable,broadbent2015delegating,mahadev2017classical}~\footnote{For more details see review of this field in \cite{fitzsimons2017private}}, verifiable delegated quantum computation~\cite{aharonov2008interactive,reichardt2012classical,fitzsimons2012unconditionally,hayashi2015verifiable,broadbent2015verify,fitzsimons2018post,takeuchi2018resource,mahadev2018classical}~\footnote{For more details see recent reviews in \cite{gheorghiu2019verification,vidick2020verifying}}, secure multiparty quantum computation \cite{kashefi2017multiparty,kashefi2017quantum,kashefi2017garbled}, quantum fully homomorphic encryption \cite{broadbent2015quantum,dulek2016quantum}. It turns out that one of the central building blocks is secure \emph{remote state preparation} ($\sf{RSP}$) that was first defined in \cite{dunjko2012blind}. At a high level, $\sf{RSP}$ resources enable a client to remotely prepare a quantum state on the server and are, therefore, the natural candidate to replace quantum channel resources in a modular fashion. These resources further appear to enable a large ecosystem of composable protocols~\cite{dunjko2012blind,dunjko2014composable}, including in particular the important \emph{Universal Blind Quantum Computation} ($\sf{UBQC}$)~\cite{broadbent2009universal} protocol used to delegate a computation to a remote quantum server who has no knowledge of the ongoing computation.

However, in most of the above-mentioned works, the users and providers do have access to quantum resources to achieve their goals, in particular to quantum channels in addition to classical communication channels. This might prove to be challenging for some quantum devices, e.g. those with superconducting qubits, and in general, it also restricts the use of these quantum cloud services to users with suitable quantum technology. Motivated by this practical constrain, \cite{CCKW18} introduced a protocol mimicking this remote state preparation resource over a purely \emph{classical} channel (under the assumption that learning with error problem is computationally hard for quantum servers). This is a cryptographic primitive between a fully classical client and a server (with a quantum computer). By the end of the interactive protocol the client has ``prepared'' remotely on the server's lab, a quantum state (typically a single qubit $\ket{+_\theta} := \frac{1}{\sqrt{2}}(\ket{0} + e^{i\theta}\ket{1})$). This protocol further enjoys some important privacy guarantees with respect to the prepared state.

The important role of such a classical $\sf{RSP}$ primitive as part of larger protocols -- most notably in their role in replacing quantum channels between client and server -- stems from their ability to make the aforementioned protocols available to classical users, in particular clients without quantum-capable infrastructure on their end. 
It is therefore of utmost importance to develop an understanding of this primitive, notably its security guarantees when composed in larger contexts such as in~\cite{gheorghiu2019computationally}.

In this paper, we initiate the study of analyzing classical remote state-preparation from first principles. We thereby follow the Constructive Cryptography (CC) framework \cite{maurer2011abstract,maurer2011constructive} to provide a clean treatment of the $\sf{RSP}$ primitive from a composable perspective. (Note that the framework is also referred to as Abstract Cryptography (AC) in earlier works.) Armed with such a definition, we then investigate the limitations and possibilities of using classical $\sf{RSP}$ both in general and in more specific contexts. Using CC is a common approach to analyze classical as well as quantum primitives and their composable security guarantees in general and in related works including~\cite{dunjko2014composable,dunjko2016blind,morimae2013composable}.

\subsection{Overview of our Contributions}
We present an informal overview of our main results. In this work, we cover the security of $\sf{RSP}_{CC}$, the class of remote state preparation protocols which only use a classical channel, and the use-case that corresponds to its arguably most important application: Universal Blind Quantum Computing ($\sf{UBQC}$) protocols with a completely classical client. More specifically, we analyze the security of $\sf{UBQC}_{CC}$, the family of protocols where a protocol in $\sf{RSP}_{CC}$ is used to replace the quantum channel from the original quantum-client $\sf{UBQC}$ protocol. An example of an $\sf{RSP}$ resource is the $\channelBB$~\footnote{The notation $\Z \frac{\pi}{2} $ denotes the set of the 4 angles $\{0,\frac{\pi}{2}, \pi, \frac{3\pi}{2} \}$.} resource (depicted in \cref{fig:cT_pi2_1}) outputting the quantum state $\Ket{+_{\theta}}$ on its right interface, and the classical description of this state, $\theta$, on its left interface.
\begin{figure}[ht]
  \centering
  \begin{bbrenv}{ctpi2}
    \begin{bbrbox}[name=$\channelBB$]%
      \pseudocode{%
        \theta \gets \left\{0, \frac{\pi}{2},\pi, \frac{3\pi}{2} \right\}
      }%
    \end{bbrbox}
    \bbrmsgspace{4mm}
    \bbrmsgfrom{top={$\theta$}}
    \bbrqryspace{4mm}
    \bbrqryto{top={$\ket{+_\theta}$}}
  \end{bbrenv}
  \caption{Ideal resource $\channelBB$}
  \label{fig:cT_pi2_1}
\end{figure}


We show in~\cref{sec:impossibleComposableRSP} a wide-ranging limitation to the universally composable guarantees that any protocol in the family $\sf{RSP}_{CC}$ can achieve. The limitation follows just from the relation between (i) the notion of classical realization and (ii) a property we call describability -- which roughly speaking measures how leaky an \RSP resource is. The limitation directly affects the amount of additional leakage on the classical description of the quantum state. In this way, it rules out a wide set of desirable resources, even against computationally bounded distinguishers.


\vspace{\baselineskip}
\noindent\textbf{Theorem~\ref{thm:nogoClassicalRSP}}\;(Security Limitations of $\sf{RSP}_{CC}$). \textit{Any $\RSP$ resource, realizable by an $\sf{RSP}_{CC}$ protocol with security against quantum polynomial-time distinguishers, must leak an encoded, but complete description of the generated quantum state to the server.}
\vspace{\baselineskip}


The importance of \cref{thm:nogoClassicalRSP} lies in the fact that it is drawing a connection between the composability of an $\sf{RSP}_{CC}$ protocol -- a \emph{computational} notion -- with the statistical leakage of the ideal functionality it is constructing -- an \emph{information-theoretic} notion. This allows us to use fundamental physical principles such as no-cloning or no-signaling in the security analysis of \emph{computationally} secure $\sf{RSP}_{CC}$ protocols. As one direct application of this powerful tool, we show that secure implementations of the ideal resource in \cref{fig:cT_pi2_1} give rise to the construction of a quantum cloner, and are hence impossible.

\begin{figure}[ht]
\centering%
\includegraphics[width=.8\linewidth]{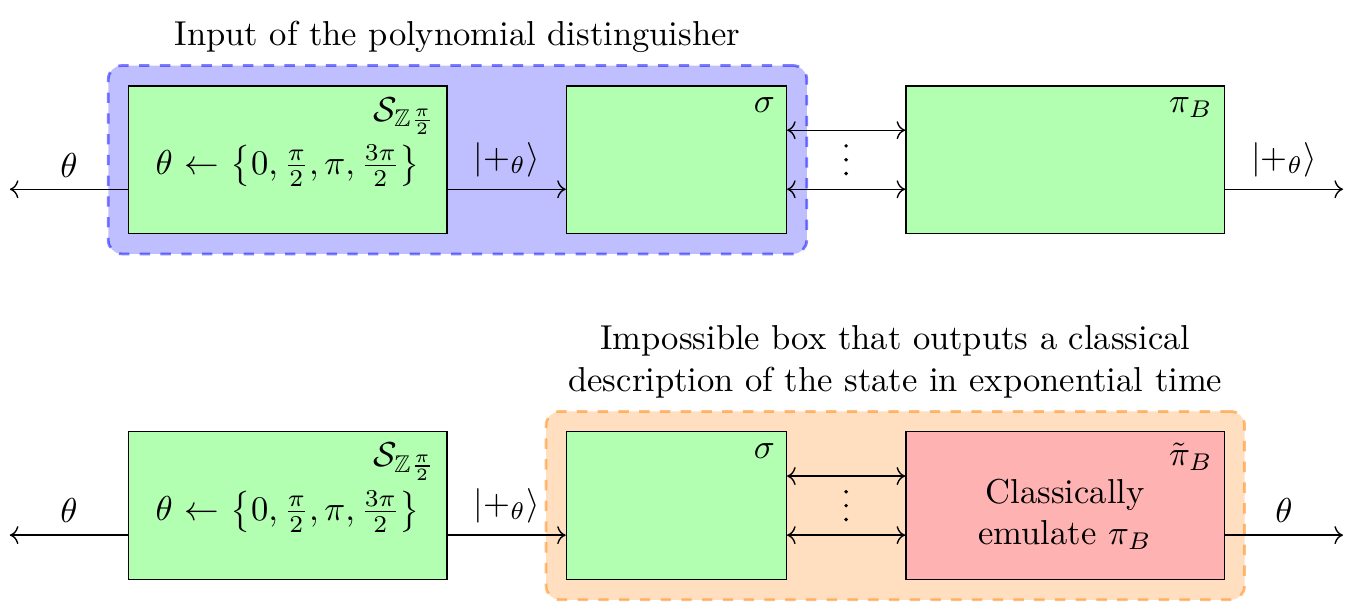}
\caption{Idea of the proof of impossibility of composable $\sf{RSP}_{CC}$, exemplified by the $\channelBB$ primitive from \cref{fig:cT_pi2}. The green boxes run in polynomial time, while the red box runs in exponential time. $\tilde{\pi}_B$ runs the same computations as $\pi_B$ by emulating it. In this way, the classical description of the quantum state can be extracted.}
\label{fig:idea_proof_nogo_RSP}
\end{figure}

\begin{proof}[Proof sketch]
While \cref{thm:nogoClassicalRSP} applies to much more general \RSP resources having arbitrary behavior at its interfaces and targeting any output quantum state, for simplicity we exemplify the main ideas of our proof for the ideal resource $\channelBB$.

The composable security of a protocol realizing $\channelBB$ implies, by definition, the existence of a simulator $\sigma$ which turns the right interface of the ideal resource into a completely classical interface as depicted in \cref{fig:idea_proof_nogo_RSP}. Running the protocol of the honest server with access to this classical interface allows the distinguisher to reconstruct the quantum state $\ket{+_\theta}$ the simulator received from the ideal resource. Since the distinguisher also has access to $\theta$ via the left interface of the ideal resource, he can perform a simple measurement to verify the consistency of the state obtained after interacting with the simulator. By the correctness of the protocol, the obtained quantum state $\ket{+_\theta}$ must therefore indeed comply with $\theta$. We emphasize that this consistency check can be performed efficiently, i.e. by \emph{polynomially-bounded} quantum distinguishers.

Since the quantum state $\ket{+_\theta}$ is transmitted from $\sigma$ to the distinguisher over a classical channel, the ensemble of exchanged classical messages must contain a complete encoding of the description of the state, $\theta$. A (possibly unbounded) algorithm can hence extract the actual description of the state by means of a classical emulation of the honest server. This property of the ideal resource is central to our proof technique, we call it \emph{describability}.
\end{proof}

Having a full description of the quantum state produced by $\channelBB$ would allow us to clone it, a procedure prohibited by the no-cloning theorem. We conclude that the resource $\channelBB$ cannot be constructed from a classical channel only.


One could attempt to modify the ideal resource, to incorporate such an extensive leakage, which is necessary as the above proof implies. However, this yields an ideal resource that is actually not a useful idealization or abstraction of the real world (because it is fully leaky) which puts in question whether they are at all useful in a composable analysis. Consider for example constructions of composite protocols that utilize the (non-leaky) ideal resource as a sub-module. These constructions require a fresh security analysis if the sub-module is replaced by any leaky version of it, but since the modified resource is very specific and must mimic its implementation (in terms of leakage) it appears that this replacement does not give any benefit compared to directly using the implementation as a subroutine and then examining the compsoable security of the combined protocol as a whole. This latter way is therefore examined next.
More precisely, we might still be able to use $\sf{RSP}_{CC}$ protocols as a subroutine in other, specific protocols, and expect the overall protocol to still construct a useful ideal functionality. The protocol family $\sf{UBQC}_{CC}$ is such an application. Unfortunately, as we show in~ \cref{sec:impossibility_composable_CUBQC}, $\sf{UBQC}_{CC}$ fails to provide the expected composable security guarantees once classical remote state preparation is used to replace the quantum channel from client to server (where composable security for UBQC refers to the goal of achieving the established ideal functionality of~\cite{dunjko2014composable} which we recall in~\cref{sec:impossibility_composable_CUBQC}). This holds even if the distinguisher is computationally bounded.


\vspace{\baselineskip}
\noindent\textbf{Theorem~\ref{thm:nogo_ubqc}}\;(Impossibility of $\sf{UBQC}_{CC}$). \textit{No $\sf{RSP}_{CC}$ protocol can replace the quantum channel in the $\sf{UBQC}$ protocol while preserving composable security.}


\begin{proof}[Proof sketch]
We first show that the existence of any composable $\sf{UBQC}_{CC}$ protocol (in the sense of achieving the ideal UBQC resource) implies the existence of a composable \emph{single-qubit} $\sf{UBQC}_{CC}$ protocol.
In turn, the impossibility of composable single-qubit $\sf{UBQC}_{CC}$ protocols is then proven in two steps.
First, we show that single-qubit $\sf{UBQC}_{CC}$ protocols can, in fact, be turned into \RSP protocols. This allows us to employ the toolbox we developed before on \RSP protocols. As a second step, we deduce that an \RSP protocol of this specific kind (that leaks the classical description, even in the form of an encoded message) would violate the no-signaling principle, thereby showing that a composable $\sf{UBQC}_{CC}$ protocol could not have existed in the first place.
\end{proof}


Finally in~\cref{sec:game_based_cubqc}, we show that the protocol family $\sf{RSP}_{CC}$ is not trivial with respect to privacy guarantees. It contains protocols with  reasonably restricted leakage that can be used as subroutines in specific applications resulting in combined protocols that offer a decent level of security. Specifically, we prove the blindness property of $\sf{QF}$-$\sf{UBQC}$, a concrete $\sf{UBQC}_{CC}$ protocol that consists of the universal blind quantum computation ($\sf{UBQC}$) protocol of \cite{broadbent2009universal} and the specific LWE-based remote state preparation ($\sf{RSP}_{CC}$) protocol from~\cite{cojocaru2019qfactory}. This yields the first provably secure $\sf{UBQC}_{CC}$ protocol from standard assumptions with a classical \RSP protocol as a subroutine.


\vspace{\baselineskip}
\noindent\textbf{Theorem~\ref{thm:game_based_security_UBQC_QF}}\;(Game-Based Security of $\sf{QF}$-$\sf{UBQC}$). \textit{The universal blind quantum computation protocol with a classical client $\sf{UBQC}_{CC}$ that combines the $\sf{RSP}_{CC}$ protocol of~\textup{\cite{cojocaru2019qfactory}} and the $\sf{UBQC}$ protocol of~\textup{\cite{broadbent2009universal}} is adaptively blind in the game-based setting. We call this protocol $\sf{QF}$-$\sf{UBQC}$. This protocol is secure under standard assumptions.}
\vspace{\baselineskip}


The statement of \cref{thm:game_based_security_UBQC_QF} can be summarized as follows: No malicious (but computationally bounded) server in the $\sf{QF}$-$\sf{UBQC}$ protocol could distinguish between two runs of the protocol performing different computations. This holds even when it is the adversary that chooses the two computations that he will be asked to distinguish. The security is achieved in the plain model, i.e., without relying on additional setup such as a measurement buffer. The protocol itself is a combination of $\sf{UBQC}$ with the QFactory protocol. For every qubit that the client would transmit to the server in the original $\sf{UBQC}$ protocol, QFactory is invoked as a subprocedure to the end of remotely preparing the respective qubit state on the server over a classical channel.

\begin{proof}[Proof sketch]
By a series of games, we show that the real protocol on a single qubit is indistinguishable from a game where the adversary guesses the outcome of a hidden coin flip. We generalize this special case to the full protocol on graphs with a polynomial number of qubits by induction over the size of the graph.
\end{proof}


\subsection{Related Work}

 While $\sf{RSP}_{CC}$ was first introduced in \cite{CCKW18} (under a different terminology), (game-based) security was only proven against weak (honest-but-curious) adversaries. Security against malicious adversaries was proven for a modified protocol in \cite{cojocaru2019qfactory}\footnote{In \cite{cojocaru2019qfactory} a verifiable version of $\sf{RSP}_{CC}$ was also given, but security was not proven in full generality.}, this protocol, called \emph{QFactory}, is the basis of the positive results in this work. In parallel \cite{gheorghiu2019computationally} gave another protocol that offers a stronger notion of \emph{verifiable} $\sf{RSP}_{CC}$ and proved the security of their primitive in the CC framework. The security analysis, however, requires an assumption of \emph{measurement buffer} resource in addition to the classical channel to construct a verifiable $\sf{RSP}_{CC}$. Our result confirms that the measurement buffer resource is a strictly non-classical assumption.  
 
In the information-theoretic setting with perfect security\footnote{By perfect security we mean at most input size is allowed to be leaked}, the question of secure delegation of quantum computation with a completely classical client was first considered in~\cite{morimae2014impossibility}. The authors showed a negative result by presenting a \emph{scheme-dependent} impossibility proof. This was further studied in~\cite{dunjko2016blind,aaronson2017implausibility} which showed that such a classical delegation would have implications in computational complexity theory. To be precise, ~\cite{aaronson2017implausibility} conjecture that such a result is unlikely by presenting an oracle separation between $\textsc{BQP}$ and the class of problems that can be classically delegated with perfect security (which is equivalent to the complexity class $\textsc{NP/poly} \cap \textsc{coNP/poly}$ as proven by \cite{abadi1987hiding}). On the other hand, a different approach to secure delegated quantum computation with a completely classical client, without going via the route of $\sf{RSP}_{CC}$, was also developed in ~\cite{mantri2017flow} where the server is unbounded and in \cite{mahadev2017classical,brakerski2018quantum} with the bounded server. The security was analysed for the overall protocol (rather than using a module to replace quantum communication). It is worth noting that~\cite{mantri2017flow} is known to be not composable secure in the Constructive Cryptography framework~\cite{atul}.
    

\section{Preliminaries}\label{sec:prelim}

 We assume basic familiarity with quantum computing, for a detailed introduction see \cite{nielsen2002quantum} (note that in this paper all Hilbert spaces are assumed to have a finite dimension). We just formalize here what we mean in this paper by \emph{quantum instrument}, which is a concept introduced by \cite{QuantumInstrument}, and which is a generalization of completely positive trace preserving (CPTP) maps to maps having both classical and quantum outputs:

\begin{definition}[Quantum Instrument]\label{def:quantumInstrument}
  A map $\Lambda: \C^{n \times n} \rightarrow \{0,1\}^{m_1} \times \C^{m_2 \times m_2}$ is said to be a quantum instrument if there exists a collection $\{\cE_y\}_{y \in \{0,1\}^{m_1}}$ of trace-non-increasing completely positive maps such that the sum is trace-preserving (i.e. for any positive operator $\rho$, $\sum_y \cE_y(\rho) = \Tr(\rho)$), and, if we define $\rho_y = \frac{\cE_y(\rho)}{\Tr(\cE_y(\rho))}$, then $\pr{\Lambda(\rho) = (y, \rho_y)} = \Tr(\cE_y(\rho))$.
\end{definition}

\subsection{The Constructive Cryptography Framework}

The Constructive Cryptography (CC) framework (also sometimes referred to as the Abstract Cryptography (AC) framework) introduced by Maurer and Renner \cite{maurer2011abstract} is a top-down and axiomatic approach, where the desired functionality is described as an (ideal) \emph{resource} $\cS$ with a certain input-output behavior independent of any particular implementation scheme. A resource has some interfaces $\cI$ corresponding to the different parties that could use the resource. In our case, we will have only two interfaces corresponding to Alice (the client) and Bob (the server), therefore $\cI = \{A, B\}$. Resources are not just used to describe the desired functionality (such as a perfect state preparation resource), but also to model the assumed resources of a protocol (e.g., a communication channel). The second important notion is the \emph{converter} which, for example, are used to define a protocol. Converters always have two interfaces, an inner and an outer one, and the inner interface can be connected to the interface of a resource. For example, if $\cR$ is a resource and $\pi_A$, $\pi_B \in \Sigma$ are two converters (corresponding to a given protocol making use of resource $\cR$) we can connect these two converters to the interface $A$ and $B$, respectively, (the resulting object being a resource as well) using the following notation: $\pi_A \cR \pi_B$.

In order to characterize the distance between two resources (and therefore the security), we use the so-called \emph{distinguishers}. We then say that two resources $\cS_1$ and $\cS_2$ are indistinguishable (within $\eps$), and denote it as $\cS_1 \approx_\eps \cS_2$, if no distinguisher can distinguish between $\cS_1$ and $\cS_2$ with an advantage greater than $\eps$. In the following, we will mostly focus on quantum-polynomial-time (QPT) distinguishers.

Central to Constructive Cryptography is the notion of a secure construction of an (ideal) resource $\cS$ from an assumed resource $\cR$ by a protocol (specified as a pair of converters).
We directly state the definition for the special case we are interested in, namely in two-party protocols between a client $A$ and a server $B$, where $A$ is always considered to be honest. The definition can therefore be simplified as follows:
\begin{definition}[See {\cite{maurer2011constructive,maurer2011abstract}}]\label{def:realize}
  Let $\cI = \{A,B\}$ be a set of two interfaces ($A$ being the left interface and $B$ the right one), and let $\cR,\cS$ be two resources. Then, we say that for the two converters $\pi_A,\pi_B$, the protocol $\protocol := (\pi_A,\pi_B)$ \emph{(securely) constructs} $\cS$ from $\cR$ within $\eps$, or that $\cR$ \emph{realizes} $\cS$ within $\eps$, denoted:
  \begin{align}
    \cR \constructs{\protocol}{\eps} \cS
  \end{align}
  if the following two conditions are satisfied:
  \begin{itemize}
  \item Availability (i.e. correctness):
  \begin{align}
    \pi_A\cR \pi_B \approx_\eps \cS \filter
  \end{align} (where $\filter$ represents a filter, i.e. a trivial converter that enforces honest/correct behavior \footnote{Usually, a filter simply sends a bit $c=0$ and then forwards all communications between its two interfaces (this filter will be denoted by $\filter^{c=0}$), but it could be a more general converter. When the filter is not clear from the context, we need to specify also which filter we consider.}, and $A \approx_\eps B$ means that no polynomial quantum distinguisher can distinguish between $A$ and $B$ (given black-box access to $A$ or $B$) with an advantage better than $\eps$)
  \item Security: there exists $\sigma \in \Sigma$ (called a simulator) such that:
  \begin{align}
    \pi_A \cR \approx_\eps \cS \sigma
  \end{align}
  \end{itemize}
  We also extend this definition when $\eps$ is a function $\eps: \N \rightarrow \R$: we say that $\cS$ is \emph{$\eps$-classically-realizable} if for any $n \in \N$, $\cS$ is $\eps(n)$-realizable\footnote{Note that here the protocols $\pi_A^{(n)}$ and $\pi_B^{(n)}$ may or may not be efficient to compute given $n$, so our nogo-result will apply to non-uniform circuits, and therefore also to uniform circuits.}.
\end{definition}

The intuition behind this definition is that if no distinguisher can know whether he is interacting with an ideal resource or with the real protocol, then it means that any attack done in the ``real world'' can also be done in the ``ideal world''. Because the ideal world is secure by definition, so is the real world. Using such a definition is particularly useful to capture the ``leakage'' of information to the server. This is quite subtle to capture in the real world, but very natural in the ideal world.

In our work, we instantiate a general model of computation to capture general quantum computations within converters which ensures that they follow the laws of quantum physics (e.g., excluding that the input-output behavior is signaling). Indeed, without such a restriction, we could not base our statements on results from quantum physics, because an arbitrary physical reality must not respect them, such as cloning of quantum states, signaling, and more. More specifically, in this work, we assume that any converter that interacts classically on its inner interface and outputs a single quantum message on its outer interface can be represented as a sequence of quantum instruments (which is a generalization of CPTP maps taking into account both quantum and classical outputs, see \cref{def:quantumInstrument}) as represented in \cref{fig:protocolintofunctions} and constitutes the most general expression of allowed quantum operations. More precisely, this model takes into account interactive converters (and models the computation in sequential dependent stages). This is similar to if one would in the classical world instantiate the converter by a sequence of classical Turing machines (passing state to each other)~\cite{Goldreich_2001}. For more details and to see why such definitions are enough to provide composability, see \cref{app:game_and_ac}.

\subsection{Notation} 

We denote by $\Z \frac{\pi}{2} $ the set of the 4 angles $\{0,\frac{\pi}{2}, \pi, \frac{3\pi}{2} \}$, and $\Z \frac{\pi}{4} = \{0, \frac{\pi}{4}, ..., \frac{7\pi}{4} \}$ the similar set of 8 angles. If $\rho$ is a quantum state, $[\rho]$ is the \emph{classical} representation (as a density matrix) of this state. We also denote the quantum state $\Ket{+_{\theta}} := \frac{1}{\sqrt{2}}(\Ket{0} + e^{i\theta} \Ket{1})$, where $\theta \in \Z \frac{\pi}{4}$, and for any angle $\theta$, $[\theta]$ will denote $[\ketbra{+_{\theta}}]$, i.e. the classical description of the density matrix corresponding to $\ket{+_\theta}$. For a protocol $\mathcal{P} = (P_1, P_2)$ with two interacting algorithms $P_1$ and $P_2$ denoting the two participating parties, let $r \gets \left\langle P_1, P_2 \right\rangle$ denote the execution of the two algorithms, exchanging messages, with output $r$. We use the notation $\cC$ to denote the \emph{classical channel} resource, that just forwards classical messages between the two parties.


\section{Impossibility of Composable Classical \texorpdfstring{$\sf{RSP}$}{RSP} \label{sec:impossibleComposableRSP}}

In this section, we first define the general notion of what \RSP tries to achieve in terms of resources and subsequently quantify information that an ideal \RSP resource must leak at its interface to the server even if the distinguisher is computationally bounded. One would expect, that against bounded distinguisher, the resource can express clear privacy guarantees, which we prove cannot be the case.

The reason is roughly as follows: assuming that there exists a simulator making the ideal resource indistinguishable from the real protocol, we can exploit this fact to construct an algorithm that can classically describe the quantum state given by the ideal resource. It is not difficult to verify that there could exist an inefficient algorithm (i.e. with exponential run-time) that achieves such a task. We show that even a computationally bounded distinguisher can distinguish the real protocol from the ideal protocol whenever a simulator's strategy is independent of the classical description of the quantum state. This would mean that for an $\sf{RSP}$ protocol to be composable there must exist a simulator that possesses at least a classical transcript encoding the description of a quantum state. This fact coupled with the quantum no-cloning theorem implies that the most meaningful and natural $\sf{RSP}$ resources cannot be realized from a classical channel alone. We finally conclude the section by looking at the class of imperfect (describable) $\sf{RSP}$ resources which avoid the no-go result at the price of being ``fully-leaky'', not standard, and having an unfortunately unclear composable security.

\subsection{Remote State Preparation and Describable Resources} \label{subsec:ccrsp_notions}

We first introduce, based on the standard definition in the Constructive Cryptography framework, the notion of \textit{correctness} and \textit{security} of a two-party protocol which constructs (realizes) a resource from a \emph{classical} channel $\cC$.

\begin{definition}[Classically-Realizable Resource]\label{def:classimplem}
  An ideal resource $\cS$ is said to be \emph{$\eps$-classically-realizable} if it is realizable (in the sense of \cref{def:realize}) from a \emph{classical} channel, i.e. if there exists a protocol $\protocol = (\pi_A, \pi_B)$ between two parties (interacting classically) such that: 
  \begin{align}
      \cC \constructs{\protocol}{\eps} \cS
  \end{align}
\end{definition}

We would like to point out that since Alice is honest, this definition incorporates already the case when Alice and Bob share purely classical resources that are achievable by Alice emulating the resource and sending Bob's output over a classical channel.

\noindent A simple ideal prototype that captures the goal of a \RSP protocol could be phrased as follows: the resource outputs a quantum state (chosen from a set of states) on one interface and classical description of that state on the other interface to the client. For our purposes, this view is too narrow and we want to generalize this notion. 
For instance, a resource could accept some inputs from the client or interact with the server and be powerful enough to comply with the above basic behavior if both follow the protocol. We would like to capture that any resource can be seen as an \RSP resource as soon as we fix a way to efficiently convert the client and server interfaces to comply with the basic prototype. To make this formal, we need to introduce some converters that will witness this: 
\begin{enumerate}
\item A converter $\cA$ will output, after interacting with the ideal resource\footnote{$\cA$ is allowed to interact with the (ideal) resource in a non-trivial manner. However, $\cA$ will often be the trivial converter in the sense that it simply forwards the output of the ideal resource, or -- when the resource waits for a simple activation input -- picks some admissible value as input to the ideal resource and forwards the obtained description to its outer interface.}, a classical description $[\rho]$ which is one of the following:
  \begin{enumerate}
  \item A density matrix (positive and with trace 1) corresponding to a quantum state~$\rho$.
  \item The null matrix, which is useful to denote the fact that we detected some deviation that should not happen in an honest run.
  \end{enumerate} 
\item A converter $\cQ$, whose goal is to output a quantum state $\rho'$ as close as possible to the state $\rho$ output by $\cA$.
\item A converter $\cP$, whose goal is to output a classical description $[\rho']$ of a quantum state $\rho'$ which is on average ``close'' to $\rho$.
\end{enumerate}
\noindent An \RSP must meet two central criteria:
\begin{enumerate}
    \item Accuracy of the classical description of the obtained quantum state: We require that the quantum state $\rho$ described by $\cA$'s output is close to $\cQ$'s output $\rho'$.
    This is to be understood in terms of the trace distance.
    \item Purity of the obtained quantum state: Since the \RSP resource aims to replace a noise-free quantum channel, it is desirable that the quantum state output by $\cQ$ admit a high degree of purity, i.e. more formally, that $\Tr \left( \rho'^2 \right)$ be close to one. Since $\rho'$ is required to be close to $\rho$, this implies a high purity of $\rho$ as well.
\end{enumerate}
It turns out that these two conditions can be unified and equivalently captured requiring that the quantity $\Tr ( \rho\rho' )$ is close to one.
A rigorous formulation of this claim and its proof is provided by \cref{lemma:equivalence_purity_closeness}.

We can also gain a more operational intuition of the notion of \RSP by considering that an \RSP resource (together with $\cA$ and $\cQ$) can be seen, not only as a box that produces a quantum state together with its description but also as a box whose accuracy can be easily \emph{tested}\footnote{This testable property will be of great importance in our argument later.}. For example, if such a box produces a state $\rho'$, and pretends that the description of that state corresponds to $\ket{\phi}$ (i.e. $[\rho] = [\ketbra{\phi}]$), then the natural way to test it would be to measure $\rho'$ by doing a projection on $\ket{\phi}$. This test would pass with probability $p_s := \braket{\phi | \rho' | \phi}$, and therefore if the box is perfectly accurate (i.e. if $\rho' = \ketbra{\phi}$), the test will always succeed. However, when $\rho'$ is far from $\ketbra{\phi}$, this test is unlikely to pass, and we will have $p_s < 1$. We can then generalise this same idea for arbitrary (eventually not pure) states by remarking that $p_s = \braket{\phi | \rho' | \phi} = \Tr(\ketbra{\phi}\rho') = \Tr(\rho\rho')$. Indeed, this last expression corresponds\footnote{Note that it also turns out to be equal to the (squared) fidelity between $\rho$ and $\rho'$ when $\rho$ is pure.} exactly to the probability of outputting $E_0$ when measuring the state $\rho'$ according to the POVM $\{E_0 := \rho, E_1 := I - \rho\}$, and since the classical description of $\rho$ is known, it is possible to perform this POVM and test the (average) accuracy of our box. This motivates the following definition for general \RSP resources.

\begin{definition}[\RSP resources]\label{def:rsp_resource}
  A resource $\cS$ is said to be a \emph{remote state preparation resource} within $\eps$ with respect to converters $\cA$ and $\cQ$ if the following three conditions hold: (1) both converters output a single message at the outer interface, where the output $[\rho]$ of $\cA$ is classical and is either a density matrix or the null matrix, and the output $\rho'$ of $\cQ$ is a quantum state; (2) the equation:
  \begin{equation}
    \esp[([\rho], \rho') \leftarrow \cA \cS \filter \cQ]{ \Tr(\rho \rho')} \geq 1-\eps
    \label{eq:remoteresource}
  \end{equation}
  is satisfied, where the probability is taken over the randomness of $\cA$, $\cS$ and $\cQ$, and finally, (3) for all the possible outputs $[\rho]$ of $([\rho], \rho') \leftarrow \cA \cS \filter \cQ$, if we define $E_0 = \rho$, $E_1 = I - \rho$, then  the POVM $\{E_0, E_1\}$ must be efficiently implementable\footnote{We could also define a similar definition when this POVM can only be approximated (for example because the distinguishers can only perform quantum circuits using a finite set of gates) and the theorems would be similar, up to this approximation, but for simplicity we will stick to that setting.} by any distinguisher.
\end{definition}

Whenever we informally speak of a resource $\cS$ as being an \RSP resource, this has to be understood always in a context where the converters $\cA$ and $\cQ$ are fixed.

\paragraph*{Describable resources.}
So far, we have specified that a resource qualifies as an \RSP resource if, when all parties follow the protocol, we know how to compute a quantum state on the right interface and classical description of a ``close'' state on the other interface. 
A security-related question now is, if it is also possible to extract (possibly inefficiently) from the right interface a \emph{classical} description of a quantum state that is close to the state described by the client. If we find a converter $\cP$ doing this, we would call the (RSP) resource \emph{describable}. The following definition captures this.

\begin{definition}[Describable Resource]\label{def:describable_resource}
  Let $\cS$ be a resource and $\cA$ a converter outputting a single classical message $[\rho]$ on its outer interface (either equal to a density matrix or the null matrix).
  Then we say that $(\cS, \cA)$ is \emph{$\eps$-describable} (or, equivalently, that $\cS$ is describable within $\eps$ with respect to $\cA$) if there exists a (possibly unbounded) converter $\cP$ (outputting a single classical message $[\rho']$ on its outer interface representing a density matrix) such that:
  \begin{equation}
    \esp[([\rho], [\rho']) \leftarrow \cA \cS \cP]{ \Tr(\rho \rho')} \geq 1 - \eps
    \label{eq:describable}
  \end{equation}
  (the expectation is taken over the randomness of $\cS$, $\cA$ and $\cP$).
\end{definition}

\paragraph*{Reproducible converters.}
In the proof of our first result, we will encounter a crucial decoding step. Roughly speaking, the core of this decoding step is the ability to convert the classical interaction with a client, which can be seen as an arbitrary encoding of a quantum state, back into an explicit representation of the state prepared by the server. The ability of such a conversion can be phrased by the following definition.

\begin{definition}[Reproducible Converter]\label{def:reproducibleconverter}
  A converter $\pi$ that outputs (on the right interface) a quantum state $\rho$ is said to be reproducible if there exists a (possibly inefficient) converter $\tilde{\pi}$ such that:
  \begin{enumerate}
  \item the outer interface of $\tilde{\pi}$ outputs only a classical message $[\rho']$
  \item the converter $\pi$ is perfectly indistinguishable from $\tilde{\pi}$ against any unbounded distinguisher $D \in \cD^u$, up to the conversion of the classical messages $[\rho']$ into a quantum state $\rho'$. More precisely, if we denote by $\cT$ the converter that takes as input on its inner interface a classical description $[\rho']$ of a quantum state and outputs that quantum state $ \rho'$ (as depicted in \cref{fig:reproducibleconverter}), we have:
    \begin{equation}
      \cC \pi \approx^{\cD^u}_0 \cC \tilde{\pi} \cT
    \end{equation}
  \end{enumerate}
\end{definition}

  \begin{figure}[ht]
    \centering
    \includegraphics[]{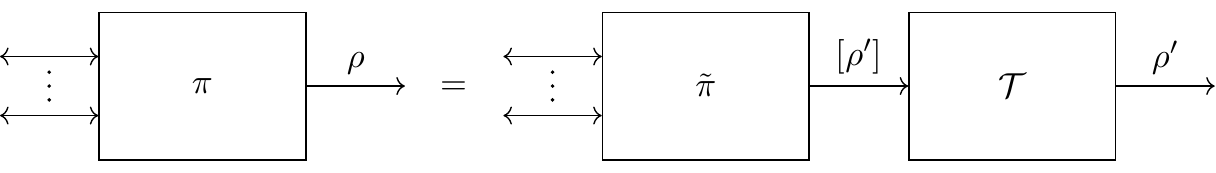}
    \caption{Reproducible converter.}
    \label{fig:reproducibleconverter}
  \end{figure}

\paragraph*{Classical communication and reproducibility.}
We see that in general, being reproducible is a property that stands in conflict with the quantum no-cloning theorem. More precisely, the ability to reproduce implies that there is a way to extract knowledge of a state sufficient to clone it. However, whenever communication is classical, quite the opposite is
true. This is formalized in the following lemma. Intuitively, it says that in the principle it is always possible to compute the exact description of the state from the classical transcript and the \emph{quantum instruments} (circuit) used to implement the action of the converter, where an instrument is a generalized CPTP map which allows a party to output both a quantum and a classical state and is formalized more precisely in \cref{def:quantumInstrument}. Recall that this is the most general way of representing a quantum operation.

  In the proof, we just need to assume that $\pi$ interacts (classically) with the inner interface first, and finally outputs a quantum state on the outer interface, so for simplicity we will stick to that setting. In this  way we can decompose $\pi$ as depicted in \cref{fig:protocolintofunctions} using the following notation:
  \begin{align}\label{eq:notation_pi_i}
    \pi := (\pi_i)_{i}
  \end{align}
  Each $\pi_i$ represents a round, and we denote with $(y_i, \rho_{i+1}) \leftarrow  \pi_i(x_i, \rho_i)$ the output of the $i$-th round, assuming that $x_i \in \{0,1\}^{l_i}$ is a classical input message sent from the inner interface, $\rho_{i}$ is the internal quantum state (density matrix) after round $i-1$, $\rho_{i+1}$ is the internal state after round $i$, and $y_i \in \{0,1\}^{l'_i} \cup \bot$ is a classical message, sent to the inner interface when $y_i \neq \bot$. For the first protocol, we set $\rho_0 = (1)$, which is the trivial density matrix of dimension 1. Moreover, when $y_i = \bot$, we do not send any message to the inner interface and instead we send $\rho_{i+1}$ to the outer interface and we stop the protocol. Note that if we want to let $\pi$ send the first message instead of receiving it, we can set $x_0 = \bot$, and similarly, if the last message is sent instead of received, we can add one more round where we set $x_{n+1} = \bot$.

  \begin{figure}[ht]
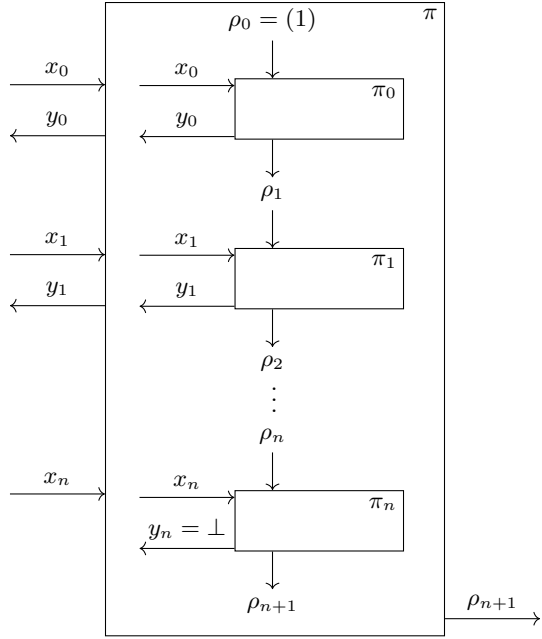

    \centering
    \begin{bbrenv}{Pi}
      \begin{bbrbox}[name=$\pi$]
        \begin{bbrenv}{Pi1}
          \begin{bbrbox}[name=$\pi_0$]
            
          \end{bbrbox}
          \bbrinput{$\rho_0 = (1)$}
          \bbroutput{$\rho_1$}
          \bbrmsgto{top=$x_0$}
          \bbrmsgfrom{top=$y_0$}
        \end{bbrenv}\\
        ~\\
        \begin{bbrenv}{Pi1}
          \begin{bbrbox}[name=$\pi_1$]
            
          \end{bbrbox}
          \bbrinput{}
          \bbrmsgto{top=$x_1$}
          \bbrmsgfrom{top=$y_1$}
          \bbroutput{$\rho_2$}
        \end{bbrenv}\\
        \vspace{-1mm}\\\hspace*{19.5mm}\vdots\vspace{5mm}\\
        \begin{bbrenv}{Pi1}
          \begin{bbrbox}[name=$\pi_n$]
            
          \end{bbrbox}
          \bbrinput{$\rho_{n}$}
          \bbrmsgto{top=$x_n$}
          \bbrmsgfrom{top={$y_n = \bot$}}
          \bbroutput{$\rho_{n+1}$}
        \end{bbrenv}\\
      \end{bbrbox}
      \bbrmsgspace{4.2mm}
      \bbrmsgto{top=$x_0$}
      \bbrmsgfrom{top=$y_0$}
      \bbrmsgspace{9mm}
      \bbrmsgto{top=$x_1$}
      \bbrmsgfrom{top=$y_1$}
      \bbrmsgspace{18.2mm}
      \bbrmsgto{top=$x_n$}
      \bbrqryspace{7.5cm}
      \bbrqryto{top=$\rho_{n+1}$}
    \end{bbrenv}
    \caption{Representation of an interactive protocol $\pi$ into a sequence of quantum instruments.}
    \label{fig:protocolintofunctions}
  \end{figure}

Now, we can prove that a party, that produces a quantum state at the end of a protocol with exclusively classical communication, is reproducible:

\begin{lemma}\label{lemma:classicalAreReproducible}
  Let $\pi = (\pi_i)_{i}$ (using the notation introduced \cref{eq:notation_pi_i}) be a converter such that:
  \begin{enumerate}
  \item it receives and sends only classical messages from the inner interfaces
  \item it outputs at the end a quantum state on the outer interface
  \item each $\pi_i$ is a quantum instrument
  \end{enumerate}
  then $\pi$ is reproducible.
\end{lemma}

\begin{proof}
  The intuition behind the proof is to argue that because the only interactions with the outside world are classical as seen from \cref{fig:protocolintofunctions}, the internal state of $\pi$ can always be computed (in exponential time) manually.\\
  More precisely, for all $i$, because $\pi_i$ is a quantum instrument, there exists a set $\{\cE_{y_i}\}$ of maps having the properties defined in \cref{def:quantumInstrument}. And because for all $y_i$, $\cE_{y_i}$ is completely positive, there exists a finite set of matrices $\{B^{(i,y_i)}_k\}_{k}$, known as Kraus operators, such that we have for all $\rho$ (and in particular for $\rho = \ket{x_i}\bra{x_i} \otimes \rho_i$):
  \begin{equation}
    \cE_{y_i}(\rho) = \sum_k B^{(i,y_i)}_k \rho B^{(i,y_i)\dagger}_k
    \end{equation}
  Therefore, for all $x_i$, $\rho_i$ and $y_i$, we have with probability $p_{y_i} := \Tr(\cE_{y_i}(\ket{x_i}\bra{x_i} \otimes \rho_i))$:
  \begin{align}
    \pi_i(x_i, \rho_i)
    &= (y_i, \cE_{y_i}(\ket{x_i}\bra{x_i} \otimes \rho_i))\\
    &= (y_i, \underbrace{\sum_k B^{(i,y_i)}_k (\ket{x_i}\bra{x_i} \otimes \rho_i) B^{(i,y_i)\dagger}_k}_{\rho_{i+1}})
    \label{eq:sigmaiplusone}
  \end{align}
  We remark that if we know $[\rho_i]$, the coefficients of the matrix $\rho_i$, then for all $y_i$ we can compute the probability $p_{y_i}$ of outputting $y_i$, and the corresponding $[\rho_{i+1}]$, (the coefficients of the matrix $\rho_{i+1}$) by just doing the above computation. So to construct $\tilde{\pi}$ (using notations from \cref{def:reproducibleconverter}) we do as follows:
  \begin{itemize}
  \item first, for all $i$ we construct $\tilde{\pi}_i$, which on input $(x_i, [\rho_i])$ outputs $(y_i, [\rho_{i+1}])$ with probability $p_{y_i}$ using the formula \cref{eq:sigmaiplusone}.
  \item then, we define $\tilde{\pi}$ as $(\tilde{\pi}_i)$ with $[\rho_0] = (1)$.
  \end{itemize}
  Then, we trivially have $\cC \pi \approx_0 \cC \tilde{\pi} \cT$, even for unbounded distinguishers, because $\tilde{\pi}$ is exactly the same as $\pi$, except that the representations of the quantum states in $\tilde{\pi}$ are matrices, while they are actual quantum states in $\pi$. Therefore, adding $\cT$ (which turns any $[\rho_i]$ into $\rho_i$) on the outer interface (which is the only interface that sends a classical state $[\rho_i]$) gives us $\pi \approx_0 \cC \tilde{\pi}\cT$.
  \end{proof}

\subsection{Classically-Realizable \texorpdfstring{$\sf{RSP}$}{RSP} are Describable}\label{subsec:imposs_cc_rsp}

In this section we show our main result about remote state preparation resources, which interestingly links a constructive notion (\emph{composability}) with respect to a computational notion with an information theoretic property (\emph{describability}).

This implies directly the \emph{impossibility result} regarding the existence of non-describable $\sf{RSP}_{CC}$ composable protocols (secure against \emph{bounded} BQP distinguishers). While this theorem does not rule out all the possible $\sf{RSP}$ resources, it shows that most ``\emph{useful}'' $\sf{RSP}$ resources are impossible. Indeed, the describable property is usually not a desirable property, as it means that an unbounded adversary could learn the description of the state he received from an ideal resource. To illustrate this theorem, we will see in the \cref{subsec:example_imposs_resource} some examples showing how this result can be used to prove the impossibility of classical protocols implementing some specific resources, and in \cref{subsec:imperfectRSP} we will see some example of ``imperfect'' resources escaping the impossibility result.

\begin{theorem}[Classically-Realizable $\sf{RSP}$ are Describable]\label{thm:nogoClassicalRSP}
 \sloppy If an ideal resource $\cS$ is both an $\eps_1$-remote state preparation with respect to some $\cA$ and $\cQ$ and $\eps_2$-classically-realizable (including against only polynomially bounded distinguishers), then it is ${(\eps_1 + 2\eps_2)}$-describable with respect to $\cA$. In particular, if $\eps_1 = \negl[n]$ and $\eps_2 = \negl[n]$, then $\cS$ is describable within a negligible error $\eps_1 + 2\eps_2 = \negl[n]$.
\end{theorem}

\begin{proof}

Let $\cS$ be an $\eps_1$-remote state preparation resource with respect to $(\cA, \cQ)$ which is $\eps_2$-classically-realizable. Then there exist $\pi_A$, $\pi_B$, $\sigma$, such that:
\begin{equation}
    \esp[([\rho], \rho') \leftarrow \cA \cS \filter \cQ]{ \Tr(\rho \rho')} \geq 1 - \eps_1
    \label{eq:complprs_remoteresource}
  \end{equation}
  \begin{equation}
    \pi_A \cC \pi_B \approx_{\eps_2} \cS \filter
    \label{eq:complprs_correct}
  \end{equation}
  and
  \begin{equation}
    \pi_A \cC \approx_{\eps_2} \cS \sigma
    \label{eq:complprs_sim}
  \end{equation}
  Now, using (\ref{eq:complprs_correct}), we get:
  \begin{equation}
    \cA \pi_A \cC \pi_B \cQ \approx_{\eps_2} \cA \cS \filter \cQ
    \label{eq:complpr}
  \end{equation}
  So it means that we can't distinguish between $\cA \cS \filter \cQ$ and $\cA \pi_A \cC \pi_B \cQ$ with an advantage better than $\eps_2$ (i.e. with probability better than $\frac{1}{2}(1+\eps_2)$). But, if we construct the following distinguisher, that runs $([\rho], \rho') \leftarrow \cA \cS \filter \cQ$, and then measures $\rho'$ using the POVM $\{E_0, E_1\}$ (possible because this POVM is assumed to be efficiently implementable by distinguishers in $\cD$), with $E_0 = [\rho]$ and $E_1 = I-[\rho]$ (which is possible because we know the classical description of $\rho$, which is positive and smaller than $I$, even when $[\rho] = 0$), we will measure $E_0$ with probability $1-\eps_1$. So it means that by replacing $\cA \cS \filter \cQ$ with $\cA \pi_A \cC \pi_B \cQ$, the overall probability of measuring $E_0$ needs to be close to $1-\eps_1$. More precisely, we need to have:
  \begin{equation}
    \esp[([\rho], \rho') \leftarrow \cA \pi_A \cC \pi_B \cQ]{ \Tr(\rho \rho')} \geq 1 - \eps_1 - \eps_2
  \end{equation}
  \begin{subproof}
    Indeed, if the above probability is smaller than $1 - \eps_1 - \eps_2$, then we can define a distinguisher that outputs $0$ if he measures $E_0$, and $1$ if he measures $E_1$, and his probability of distinguishing the two distributions would be equal to:
    \begin{align}
      &\frac{1}{2} \esp[([\rho], \rho') \leftarrow \cA \cS \filter \cQ]{\Tr(\rho\rho')} + \frac{1}{2} \esp[([\rho], \rho') \leftarrow \cA \pi_A \cC \pi_B \cQ]{\Tr((I-\rho)\rho')}\\
      &> \frac{1}{2} \left( (1-\eps_1) + 1 - (1-\eps_1-\eps_2) \right)\\
      &= \frac{1}{2} (1 + \eps_2)
    \end{align}
    So this distinguisher would have an advantage greater than $\eps_2$, which is in contradiction with \cref{eq:complpr}.
  \end{subproof}
  Using a similar argument and \cref{eq:complprs_correct}, we have:
  \begin{equation}
    \esp[([\rho], \rho') \leftarrow \cA \cS \sigma \pi_B \cQ]{ \Tr(\rho \rho')} \geq 1 - \eps_1 - 2\eps_2
  \end{equation}
  We will now use $\pi_B \cQ$ to construct a $\cB$ that can describe the state given by the ideal resource. To do that, because $\pi_B \cQ$ interacts only classically with the inner interface and outputs a single quantum state on the outer interface, then according to \cref{lemma:classicalAreReproducible}, $\pi_B \cQ$ is reproducible, i.e. there exists\footnote{Note that here $\cB$ is not efficient anymore, so that's why in the describable definition we don't put any bound on $\cB$, but of course the proof does apply when the distinguisher is polynomially bounded.} $\cB$ such that $\cC \pi_B \cQ \approx_0 \cC \cB \cT$. Therefore\footnote{Indeed, we also have in particular $\cA \cS \sigma \cC \pi_B \cQ \approx_0 \cA \cS \sigma \cC \cB \cT$, and because $\cC$ is a neutral resource \cite[Sec. C.2]{maurer2011abstract} we can remove $\cC$.}, we have:
  \begin{equation}
    \esp[([\rho], \rho') \leftarrow \cA \cS \sigma \cB \cT]{ \Tr(\rho \rho')} \geq 1 - \eps_1 - 2\eps_2
  \end{equation}
  But because $\cT$ simply converts the classical description $[\rho']$ into $\rho'$, we also have:
  \begin{equation}
    \esp[([\rho], [\rho']) \leftarrow \cA \cS \sigma \cB]{ \Tr(\rho \rho')} \geq 1 - \eps_1 - 2\eps_2
  \end{equation}
  After defining $\cP = \sigma \cB$, we have that $\cS$ is $(\eps_1 + 2\eps_2)$-describable, which ends the proof.
\end{proof}

\subsection{\texorpdfstring{$\sf{RSP}$}{RSP} Resources Impossible to Realize Classically
\label{subsec:example_imposs_resource}}

In the last section we proved that if an $\sf{RSP}$ functionality is classically-realizable (secure against polynomial quantum distinguishers), then this resource is describable by an unbounded adversary having access to the right interface of that resource.

Our main result in the previous section directly implies that as soon as there exists \emph{no unbounded} adversary that, given access to the right interface, can find the classical description given on the left interface, then the \RSP resource is \emph{impossible} to classically realize (against \emph{bounded} BQP distinguishers). Very importantly, this no-go result shows that the \emph{only} type of $\sf{RSP}$ resources that can be classically realized are the ones that \emph{leak} on the right interface enough information to allow an (possibly unbounded) adversary to determine the classical description given on the left interface. From a security point of view, this property is highly non-desirable, as the resource must leak the \emph{secret description} of the state at least in \emph{some representation}.

In this section we present some of these $\sf{RSP}$ resources that are impossible to classically realize. 

\begin{definition}[Ideal Resource $\channelBB$]\label{def:s_bb84}
$\channelBB$ is the verifiable $\sf{RSP}$ resource ($\sf{RSP}$ which does not allow any deviation from the server), that receives no input, that internally picks a random $\theta \leftarrow \Z \frac{\pi}{2}$, and that sends $\theta$ on the left interface, and $\ket{+_\theta}$ on the right interface as shown in \cref{fig:cT_pi2}.
\end{definition}

\begin{figure}[ht]
  \centering
  \begin{bbrenv}{ctpi2}
    \begin{bbrbox}[name=$\channelBB$]%
      \pseudocode{%
        \theta \gets \Z \frac{\pi}{2}
      }%
    \end{bbrbox}
    \bbrmsgspace{4mm}
    \bbrmsgfrom{top={$\theta$}}
    \bbrqryspace{4mm}
    \bbrqryto{top={$\ket{+_\theta}$}}
  \end{bbrenv}
  \caption{Ideal resource $\channelBB$}
  \label{fig:cT_pi2}
\end{figure}

\begin{lemma}
  There exists a universal constant $\eta > 0$, such that for all $0 \leq \eps < \eta$ the resource $\channelBB$ is not $\eps$-classically-realizable.
\end{lemma}

\begin{proof}
  This proof is at its core a direct consequence of quantum no-cloning: If we define $\cA(\theta) := [\ketbra{+_\theta}]$ ($\cA$ just converts $\theta$ into its classical density matrix representation) and $\cQ$ the trivial converter that just forwards any message, then $\channelBB$ is a $0$-remote state preparation resource with respect to $\cA$ and $\cQ$ because:
  \begin{align}
    \esp[([\rho], \rho') \leftarrow \cA \channelBB \filter \cQ]{ \Tr(\rho \rho')} &= \frac{1}{4} \sum_{\theta \in \Z \frac{\pi}{2}} \Tr(\ketbra{+_\theta}\ketbra{+_\theta}) = 1 \geq 1 - 0
  \end{align}
  Then, we remark also that there exists a constant $\eta > 0$ such that for all $\delta < \eta$, $\channelBB$ is not $\delta$-describable with respect to $\cA$.
  \begin{subproof}
    Indeed, it is first easy to see that $\channelBB$ is not $0$-describable with respect to $\cA$. Indeed, we can assume by contradiction that there exists $\cP$ such that:
    \begin{align}
      \esp[([\rho], [\rho']) \leftarrow \cA \channelBB \cP]{ \Tr(\rho \rho')} = 1
      \label{eq:preq1}
    \end{align}
    Then, because $\rho = \ketbra{+_\theta}$ is a pure state, $\Tr(\rho \rho')$ corresponds to the fidelity of $\rho$ and $\rho'$, so $\Tr(\rho \rho') = 1 \Leftrightarrow \rho = \rho'$. But this is impossible because $\cP$ just has a quantum state $\rho$ as input, and if he can completely describe this quantum state then he can actually clone perfectly the input state with probability 1. But because the different possible values of $\rho$ are not orthogonal, this is impossible due to the no-cloning theorem.

    Moreover, it is also not possible to find a sequence $(\cP^{(n)})_{n \in N}$ of CPTP maps that produces two copies of $\rho$ with a fidelity arbitrary close to 1 (when $n \rightarrow \infty$), because CPTP maps are compact and the fidelity is continuous. \\ Therefore, there exists a constant $\eta > 0$,\footnote{Note that for finding a more precise bound for $\eta$, it is possible to use Semidefinite Programming (SDP), or the method presented in \cite[p. 2]{KRK2012}. However in our case it is enough to say that $\eps > 0$ as we are interested only in asymptotic security.} such that:
    \begin{equation}
        \esp[([\rho], [\rho']) \leftarrow \cA \channelBB \cP]{ \Tr(\rho \rho')} < 1 - \eta
    \end{equation}
  \end{subproof}
  \sloppy Now, by contradiction, we assume that $\channelBB$ is $\eps$-classically-realizable. Because ${\lim_{n \rightarrow \infty} \eps(n) = 0}$, there exists $N \in \N$ such that $\eps(N) < \eta/2$. So, using \cref{thm:nogoClassicalRSP}, $\channelBB$ is $2\eps(N)$-describable with respect to $\cA$, which contradicts $2\eps(N) < \eta$.
\end{proof}

Next, we describe $\sf{RSP}_V$, a variant of $\channelBB$ introduced in \cite{gheorghiu2019computationally}. In the latter, $\sf{RSP}_V$, the adversary can make the resource abort, that the set of output states is bigger, and that the client can partially choose the basis of the output state. Similar to the $\channelBB$, we prove that classically-realizable $\sf{RSP}_V$ is not possible. Before going into the details of the no-go result, we formalize the ideal resource for a verifiable remote state preparation, $\sf{RSP}_V$, below.  

\begin{definition}[Ideal Resource $\sf{RSP}_V$, See \cite{gheorghiu2019computationally}]\label{def:s_rsp_v}
The ideal verifiable remote state preparation resource, $\sf{RSP}_V$, takes an input $W \in \{X,Z\}$ on the left interface, but no honest input on the right interface. The right interface has a filtered functionality that corresponds to a bit $c \in \{0,1\}$. When $c=1$, $\sf{RSP}_V$ outputs error message $\sf{ERR}$ on both the interfaces, otherwise:
 \begin{enumerate}
     \item if $W = Z$ the resource picks a random bit $b$ and outputs $b \in \Z_2$ to the left interface and a computational basis state $\ket{b}\bra{b}$ to the right interface;
     \item if $W = X$ the resource picks a random angle $\theta \in \Z\frac{\pi}{4}$ and outputs $ \theta$ to the left interface and a quantum state $ \ket{+_\theta}\bra{+_\theta}$ to the right interface.
 \end{enumerate}
\end{definition}

\begin{corollary}
    There exists a universal constant $\eta > 0$, such that for all $0 \leq \eps < \eta$ the resource $\sf{RSP}_V$ is not $\eps$-classically-realizable.
\end{corollary}
\begin{proof}
The proof is quite similar to the proof of impossibility of $\channelBB$. The main difference is that we need to address properly the abort case when $c=1$. The main idea is to define $\cA$ a bit differently: $\cA$ picks always $W = X$, and outputs as $\rho$ the classical density matrix corresponding to $s$ when $s \neq \ERR$, and when $s = \ERR$, $\cA$ outputs the null matrix $\rho = 0$ ($\cQ$ is still the trivial converter). It is easy to see again that this resource is a $0$-remote state preparation resource, and it is also impossible to describe it with arbitrary small probability: indeed, when $c = 1$, $\rho = 0$, so the trace $\Tr(\rho\rho')$ (that appears in \cref{eq:describable}) is equal to 0. Therefore, from a converter $\cP$ that (sometimes) inputs $c=1$, we can always increase the value of $\Tr(\rho\rho')$ by creating a new converter $\cP'$ turning $c$ into 0. And we are basically back to the same picture as $\channelBB$, where we have a set of states that is impossible to clone with arbitrary small probability, which finishes the impossibility proof.
\end{proof}
\begin{remark}
  Note that our impossibility of classically-realizing $\sf{RSP}_V$ does not contradict the result of \cite{gheorghiu2019computationally}.
  Specifically, in their work they make use of an additional assumption (the so called ``Measurement Buffer'' resource), which ``externalizes'' the measurement done by the distinguisher onto the simulator. In practice, this allows the simulator to change the state on the distinguisher side without letting him know. However, what our result shows is that it is impossible to realize this Measurement Buffer resource with a protocol interacting purely classically. Intuitively, the Measurement Buffer re-creates a quantum channel between the simulator and the server: when the simulator is not testing that the server is honest, the simulator replaces the state of the server with the quantum state sent by the ideal resource. This method has however a second drawback: it is possible for the server to put a known state as the input of the Measurement Buffer, and if he is not tested on that run (occurring with probability $\frac{1}{n}$), then he can check that the state has not been changed, leading to polynomial security (a polynomially bounded distinguisher can distinguish between the ideal and the real world). And because in CC, the security of the whole protocol is the sum of the security of the inner protocols, any protocol using this inner protocol will not be secure against polynomial distinguishers.
\end{remark}

\subsection{Accepting the Limitations: Fully Leaky \texorpdfstring{$\sf{RSP}$}{RSP} resources}\label{subsec:imperfectRSP}

As explained in the previous section, \cref{thm:nogoClassicalRSP} rules out all resources that are impossible to be \emph{describable} with unbounded power, and that the only type of classically-realizable \RSP resources would be the one leaking the full classical description of the output quantum state to an unbounded adversary, which we will refer to as being \emph{fully-leaky} $\sf{RSP}$. 
\noindent Fully-leaky $\sf{RSP}$ resources can be separated into two categories:
\begin{enumerate}
    \item If the $\sf{RSP}$ is describable in quantum polynomial time, then the adversary can get the secret in polynomial time. This is obviously not an interesting case as the useful properties that we know from quantum computations (such as UBQC) cannot be preserved if such a resource is employed to prepare the quantum states.
    \item If the $\sf{RSP}$ are only describable using unbounded power, then these \emph{fully-leaky} $\sf{RSP}$ resources are not trivially insecure, but their universally composable security remains unclear. Indeed, it defeats the purpose of aiming at a nice ideal resource where the provided security should be clear ``by definition'' and it becomes hard to quantify how the additional leakage could be used when composed with other protocols. A possible remedy would be to show restricted composition following~\cite{DBLP:journals/iacr/JostM17} which we discuss at the end of this paragraph.
\end{enumerate}
For completeness, we present an example of a resource that stands in this second category when assuming that post-quantum encryption schemes exist (e.g. based on the hardness of the LWE problem). As explained before, this resource needs to completely leak the description of the classical state, which in our case, is done by leaking an encryption of the description of the output state. The security guarantees therefore rely on the properties of the encryption scheme, and not on an ideal privacy guarantee as one would wish for, which is an obvious limitation.

  \paragraph*{\bf A concrete example.} In this section we focus on the second category of fully-leaky $\sf{RSP}$ and we show an example of resource that belongs to this class and a protocol realizing this resource. The fully-leaky $\sf{RSP}$ resource that we will implement, produces a BB84 state (corresponding to the set of states produced by the simpler QFactory protocol) and is described below: 
  
  \begin{definition}[Ideal Resource $\RSP^{4-states, \cF}_{CC}$]
    Let $\cF = (\Gen, \Enc, \Dec)$ be a family of public-key encryption functions. Then, we define $\RSP^{4-states, \cF}_{CC}$ as pictured in~\cref{fig:RSPenc_cC_BB84}. $B_1$ represents the basis of the output state, and is guaranteed to be random even if the right interface is malicious. $B_2$ represents the value bit of the output state when encoded in the basis $B_1$, and in the worst case it can be chosen by the right interface in a malicious scenario\footnote{Note that here the right interface can have (in a malicious scenario) full control over $B_2$, but in the QFactory Protocol~\ref{protocol:qfactoryReal} it is not clear what an adversary can do concerning $B_2$.}. Note however that in a malicious run, the adversary does not have access (at least not directly from the ideal resource) to the quantum state whose classical description is known by the classical client.
    \begin{figure}[ht]
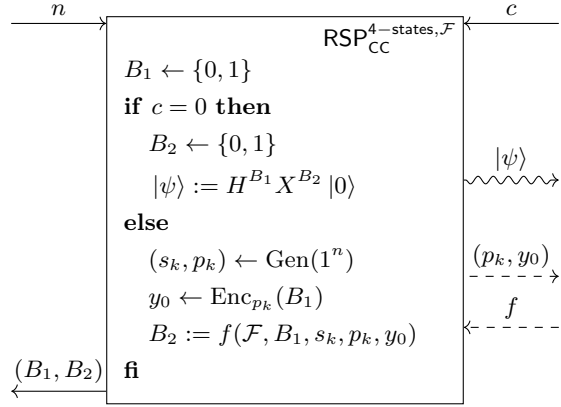

      \centering
      \begin{bbrenv}{SBQC}
        \begin{bbrbox}[name=$\RSP^{4-states, \cF}_{CC}$]%
          \pseudocode{%
            B_1 \gets \{0,1\}\\
            \pcif c=0 \pcthen\\%
            \t B_2 \gets \{0,1\}\\
            \t \ket{\psi} := H^{B_1}X^{B_2}\ket{0}\\
            \pcelse\\
            \t (s_k, p_k) \gets \Gen(1^n)\\
            \t y_0 \leftarrow \Enc_{p_k}(B_1)\\
            \t B_2 := f(\cF, B_1, s_k, p_k, y_0)\\
            \pcfi
          }%
        \end{bbrbox}
        \bbrmsgto{top={$n$}}
        \bbrmsgspace{42mm}
        \bbrmsgfrom{top={$(B_1, B_2)$}}
        \bbrqryfrom{top=$c$}
        \bbrqryspace{14mm}
        \bbrqryto{top={$\ket{\psi}$}, edgestyle={decorate, decoration={snake, segment length=2mm, amplitude=.5mm, pre length=1mm}}}
        \bbrqryspace{6mm}
        \bbrqryto{top={$(p_k, y_0)$}, edgestyle={dashed}}
        \bbrqryfrom{top={$f$}, edgestyle={dashed}}
      \end{bbrenv}
      \caption{Ideal resource $\RSP^{4-states, \cF}_{CC}$, which prepares atm{one of the four BB84 states}. The ``snake'' arrow is sent only in the honest case ($c=0$), and the dashed arrows are send/received only in the malicious case ($c=1$).}
      \label{fig:RSPenc_cC_BB84}
    \end{figure}
  \end{definition}

  \begin{lemma}\label{lemma:qfactorysecurewithleakage}
    The 4-states QFactory protocol \textup{\cite{cojocaru2019qfactory}} (Protocol~\textup{\ref{protocol:qfactoryReal}}) securely constructs $\RSP^{4-states, \cF}_{CC}$ from a classical channel, where $\cF$ is defined as follows:
    \begin{enumerate}
    \item $(t_K, K) \leftarrow \Gen(1^n)$ outputs two matrices: public $K$ (used to describe the function) and secret $t_K$ (a trapdoor used to invert the function) as defined in \textup{\cite{cojocaru2019qfactory,CCKW18}} (which is itself based on the learning with errors problem and the construction presented in \textup{\cite{MP11}});
    \item $y_0 \leftarrow \Enc_K(B_1)$ , where $y_0 = Ks_0 + e_0 + B_1 \begin{pmatrix}
        q/2 & 0 & \dots & 0
      \end{pmatrix}^T$, $s_0$ and $e_0$ being sampled accordingly to some distribution presented in \textup{\cite{cojocaru2019qfactory,CCKW18}}
    \item $B_1 \leftarrow \Dec_{t_K}(y)$ - using $t_K$ we can efficiently obtain $B_1$ from $y_0$.
    \end{enumerate}
  \end{lemma}
  
  \begin{proof}
    We already know that the protocol of QFactory $(\pi_A,\pi_B)$ is correct with super-polynomial probability if the parameters are chosen accordingly (\cref{thm:correctness}), therefore
    \begin{equation}
        \pi_A \cC \pi_B \approx_\eps \RSP^{4-states, \cF}_{CC} \filter
    \end{equation} for some negligible $\eps$. We now need to find a simulator $\sigma$ such that 
    \begin{equation}
    \pi_A \cC \approx_{\eps'} {\sf{RSP}_{CC}^{\text{4-states}, \mathcal{F}}} \sigma
    \end{equation}
 The simulator is trivial here: it sends $c = 1$ to ideal resource  then, it just forwards the $(K, y_0)$ given by the resource to its outer interface, and when it receives the $(y,b)$ corresponding to the measurements performed by the server, it just sets the deviation $f$ to be the same function as the one computed by $\pi_A$. Therefore, $\pi_A \cC \approx_0 {\sf{RSP}_{CC}^{\text{4-states}, \mathcal{F}}} \sigma $, which ends the proof.
  \end{proof}

 \paragraph*{\bf Concluding remarks.}
 We see that using this kind of leaky resource is not desirable: the resources are non-standard and it seems hard to write a modular protocol with this resource as an assumed resource. The resource is very specific and mimics its implementation. As such, we cannot really judge its security. 
  
  On the other hand however, if a higher-level protocol did guarantee that the value $B_2$ always remains hidden, i.e., a higher level protocol's output does not depend on on $B_2$ (e.g., by blinding it all the time), it is easy to see that we could simulate $y_0$ without knowledge about $B_1$ thanks to the semantic security of the encryption scheme. If we fix this restricted context, the ideal resource in~\cref{fig:RSPenc_cC_BB84} could be re-designed to not produce the output $(p_k,y_0)$ at all and therefore, by definition, leak nothing extra about the quantum state (note that in such a restricted context, the simulator can simply come up with a fake encryption that is indistinguishable). This can be made formal following ~\cite{DBLP:journals/iacr/JostM17}. We note in passing that this particular example quite severely restricts applicability unfortunately. 
  Indeed, it is interesting future research whether it is possible to come up with restricted yet useful contexts that admit nice ideal resources for \RSP following the framework in~\cite{DBLP:journals/iacr/JostM17}.


\section{Impossibility of Composable Classical-Client \texorpdfstring{$\sf{UBQC}$}{UBQC}} \label{sec:impossibility_composable_CUBQC}
In the previous section, we showed that it was impossible to get a (useful) composable $\sf{RSP}_{CC}$ protocol. A (weaker) \RSP protocol, however, could still be used internally in other protocols, hoping for the overall protocol to be composably secure. To this end, we analyze the composable security of a well-known delegated quantum computing protocol, universal blind quantum computation ($\sf{UBQC}$), proposed in ~\cite{broadbent2009universal}. The $\sf{UBQC}$ protocol allows a semi-quantum client, Alice, to delegate an arbitrary quantum computation to a (universal) quantum server Bob, in such a way that her input, the quantum computation and the output of the computation are information-theoretically hidden from Bob. The protocol requires Alice to be able to prepare single qubits of the form $\Ket{+_{\theta}}$, where $\theta \in \Z \frac{\pi}{4}$ and send these states to Bob at the beginning of the protocol, the rest of the communication between the two parties being classical. We define the family of protocols $\sf{RSP}^{8-states}_{CC}$ as the $\sf{RSP}$ protocols that classically delegate the preparation of an output state $\Ket{+_{\theta}}$, where $\theta \in \Z \frac{\pi}{4}$. That is, without loss of generality, we assume a pair of converters $P_A$, $P_B$ such that the resource $R:= P_A \cC P_B$ has the behavior of the prototype \RSP resource except with negligible probability. Put differently, we assume we have an (except with negligible error) \emph{correct} RSP protocol, but we make \emph{no assumption about the security} of this protocol.
Therefore, one can directly instantiate the quantum interaction with the $\sf{RSP}^{8-states}_{CC}$ at the first step as shown in~Protocol \ref{protocol:ubqcReal}. While $\sf{UBQC}$ allows for both quantum and classical outputs and inputs, given that we want to remove the quantum interaction in favor of a completely classical interaction, we only focus on the classical input and classical output functionality of $\sf{UBQC}$ in the remaining of the paper. \\ \\ \\

\vspace{\baselineskip}
\begin{breakablealgorithm}
\caption{$\sf{UBQC}$ with  $\sf{RSP}^{8-states}_{CC}$ (See \cite{broadbent2009universal})} \label{protocol:ubqcReal}
\begin{itemize}
    \item \textbf{Client's classical input:} An $n$-qubit unitary $U$ that is represented as set of angles $\{\phi\}_{i, j}$ of a one-way quantum computation over a brickwork state/cluster state~\cite{mantri2017universality}, of the size $n \times m$, along with the dependencies X and Z obtained via flow construction~\cite{danos2006determinism}.
\item \textbf{Client's classical output:} The measurement outcome $\bar{s}$ corresponding to the $n$-qubit quantum state, where  $\bar{s} = \bra{0}U\ket{0}$.    
\end{itemize}
\begin{enumerate}
\item Client and Server runs $n \times m$ 
different instances of $\sf{RSP}^{8-states}_{CC}$ (in parallel) to obtain $\theta_{i, j}$ on client's side and $\ket{+_{{\theta}_{i, j}}}$ on server's side, where $\theta_{i, j} \gets \Z \frac{\pi}{4}$, $i \in \{1, \cdots , n\}$, $j \in \{1, \cdots , m\}$
\item Server entangles all the qubits, $n \times (m-1)$ received from $\sf{RSP}^{8-states}_{CC}$, by applying controlled-Z gates between them in order to create a graph state $\mathcal{G}_{n \times m}$ 
\item For $j \in [1,m]$ and $i \in [1,n]$
\begin{enumerate}
\item Client computes $\delta_{i, j} = \phi_{i, j}' + \theta_{i, j} + r_{i, j}\pi$, $r_{i, j} \gets \{0, 1\}$, where $\phi_{i, j}' = (-1)^{s_{i, j}^X} \phi_{i, j} + s_{i, j}^Z \pi$ and $s_{i, j}^X$ and $s_{i, j}^Z$ are computed using the previous measurement outcomes and the X and Z dependency sets. Client then sends the measurement angle $\delta_{i, j}$ to the Server.
\item Server measures the qubit $\ket{+_{{\theta}_{i, j}}}$ in the basis $\{\ket{+_{{\delta}_{i, j}}}, \ket{-_{{\delta}_{i, j}}}\}$ and obtains a measurement outcome $s_{i, j} \in \{0, 1\}$. Server sends the measurement result to the client.
\item Client computes $\bar{s}_{i, j} = s_{i, j} \oplus r_{i, j}$.
\end{enumerate}
\item The measurement outcome corresponding to the last layer of the graph state ($j = m$) is the outcome of the computation.
\end{enumerate}
\end{breakablealgorithm}
\vspace{\baselineskip}
Note that Protocol~\ref{protocol:ubqcReal} is based on measurement-based model of quantum computing (MBQC). This model is known to be equivalent to the quantum circuit (up to polynomial overhead in resources) and does not require one to perform quantum gates on their side to realize arbitrary quantum computation. Instead, the computation is performed by an (adaptive) sequence of single-qubit projective measurements that steer the information flow across a highly entangled resource state. Intuitively, $\sf{UBQC}$ can be seen as a distributed MBQC where the measurements are performed by the server whereas the classical update of measurement bases is perfomed by the client. Since the projective measurements in quantum physics, in general, are probabilistic in nature and therefore, the client needs to update the measurement bases (and classically inform the server about the update) based on the outcomes of the earlier measurements to ensure the correctness of the computation. Roughly speaking, this information flow is captured by the X and Z dependencies. For more details, we refer the reader to~\cite{raussendorf2001one,nielsen2006cluster}.

Next, we show that the Universal Blind Quantum Computing protocol~\cite{broadbent2009universal}, which is proven to be secure in the Constructive Cryptography framework \cite{dunjko2014composable}, cannot be proven composably secure (for the same ideal resource) when the quantum interaction is replaced with $\sf{RSP}_{CC}$ (this class of protocol is denoted as $\sf{UBQC}_{CC}$). We also give an outlook that the impossibility proof also rules out weaker ideal resources.

\subsection{Impossibility of Composable \texorpdfstring{$\sf{UBQC}_{CC}$}{UBQCcc} on 1 Qubit \label{subsec:imposs_single_cc_ubqc}}

In order to prove that there exists no $\sf{UBQC}_{CC}$ protocol, we will first focus on the simpler case when the computation is described by a single measurement angle. The resource that performs a blind quantum computation on one qubit ($\cS_{UBQC1}$) is defined as below:

\begin{definition}[Ideal resource of single-qubit $\sf{UBQC}$\label{def:ideal_bcq1} (See~\cite{dunjko2014composable})] The definition of the ideal resource $\cS_{UBQC1}$, depicted in \cref{fig:SBQC1}, achieves blind quantum computation specified by a single angle $\phi$. The input $(\xi, \rho)$ is filtered when $c=0$. The $\xi$ can be any deviation (specified for example using the classical description of a CPTP map) that outputs a classical bit, and which can depend on the computation angle $\phi$ and on some arbitrary quantum state $\rho$.
  \begin{figure}[ht]
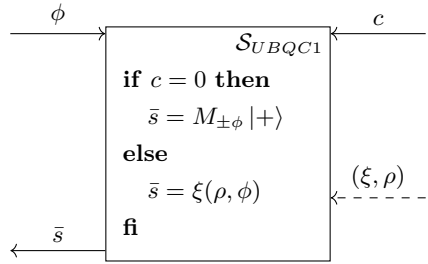

    \centering
    \begin{bbrenv}{SBQC}
      \begin{bbrbox}[name=$\cS_{UBQC1}$]%
        \pseudocode{%
          \pcif c=0 \pcthen\\%
          \t \bar{s} = M_{\pm\phi}\ket{+}\\
          \pcelse\\
          \t \bar{s} = \xi(\rho, \phi)\\
          \pcfi
        }%
      \end{bbrbox}
      \bbrmsgto{top=$\phi$}
      \bbrmsgspace{22mm}
      \bbrmsgfrom{top={$\bar{s}$}}
      \bbrqryfrom{top=$c$}
      \bbrqryspace{15mm}
      \bbrqryfrom{top={$(\xi,\rho)$}, edgestyle={dashed}}
    \end{bbrenv}
    \caption{Ideal resource $\cS_{UBQC1}$ for $\sf{UBQC}$ with one angle, with a filtered (dashed) input. In the case of honest server the output $\bar{s} \in \{0,1\}$ is computed by measuring the qubits $\ket{+}$ in the $\{\ket{+_{\phi}}, \ket{-_{\phi}}\}$ basis. On the other hand if $c = 1$ any malicious behaviour of server can be captured by $(\xi, \rho)$, i.e. the output $\bar{s}$ is computed by applying the CPTP map $\xi$ on the input $\phi$ and on another auxiliary state $\rho$ chosen by the server.}
    \label{fig:SBQC1}
  \end{figure}
\end{definition}

\begin{theorem}[No-go composable classical-client single-qubit $\sf{UBQC}$]\label{thm:nogo_ubqc1}
\sloppy  Let $(P_A, P_B)$ be a protocol interacting only through a classical channel $\cC$, such that ${(\theta, \rho_B) \leftarrow (P_A \cC P_B)}$ with $\theta \in \Z \frac{\pi}{4}$, and such that (by correctness) the trace distance between $\rho_B$ and $\ket{+_\theta}\bra{+_\theta}$ is negligible with overwhelming probability\footnote{In the following, the parties $P_A$ and $P_B$ (and therefore $\pi_A$ and $\pi_B$) and the simulator $\sigma$ depend on some security parameter $n$, but, in order to simplify the notations and the proof, this dependence will be implicit. We are as usual interested only in the asymptotic security, when $n \rightarrow \infty$.} with overwhelming probability\footnote{Note that here $\rho_B$ is different at every run: it corresponds to the density matrix of the state obtained after running $P_B$, when tracing out the environment and the internal registers of $P_B$ and $P_A$.}. Then, if we define $\pi_A$ and $\pi_B$ as the $\sf{UBQC}$ protocol on one qubit that makes use of $(P_A, P_B)$ as a sub-protocol to replace the quantum channel (as pictured in \cref{fig:honestubqc}), $(\pi_A, \pi_B)$ is not composable, i.e. there exists no simulator $\sigma$ such that:
  \begin{align}
    \pi_A \cC \pi_B &\approx_\eps \cS_{UBQC1} \filter^{c=0} \label{eq:correctnessrsp1}\\
    \pi_A \cC &\approx_\eps \cS_{UBQC1} \sigma \label{eq:soundnessrsp1}
  \end{align}
  for some negligible $\eps = \negl[n]$.
  \begin{figure}[ht]
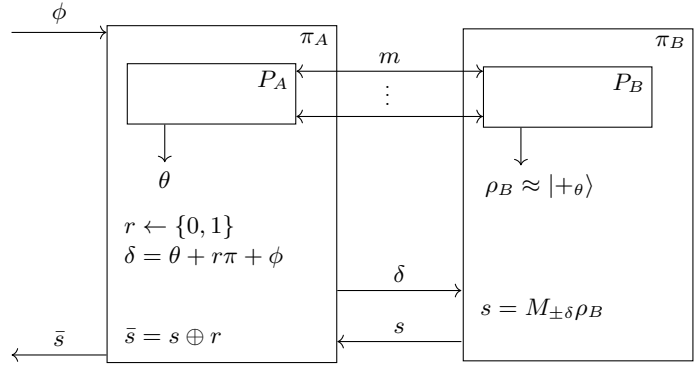

    \centering
    \begin{bbrenv}{A}
      \begin{bbrbox}[name=$\pi_A$]
        \begin{bbrenv}{RSPA}
          \begin{bbrbox}[name=$P_A$]
          \end{bbrbox}
          \bbroutput{$\theta$}
        \end{bbrenv}\\
        \vspace{.5em}\\
        $r \gets \{0,1\}$\\
        $\delta = \theta + r\pi + \phi$\\
        \vspace{3mm}\\
        $\bar{s} = s \xor r$
      \end{bbrbox}%
      \bbrmsgto{top=$\phi$}
      \bbrmsgspace{36mm}
      \bbrmsgfrom{top=$\bar{s}$}
      \begin{bbroracle}{B}
        \begin{bbrbox}[name=$\pi_B$]
          \begin{bbrenv}{RSPB}
            \begin{bbrbox}[name=$P_B$]
            \end{bbrbox}
            \bbroutput{\hspace*{5mm}$\rho_B \approx \ket{+_\theta}$}
          \end{bbrenv}\\
          \vspace{12mm}\\
          $s = M_{\pm \delta}\rho_B$\\
          \pcdraw{
            \draw[<->] ($(RSPA.east)+(0mm,3mm)$) -- node[midway] (m) {} node[midway, above, inner sep=1mm]{$m$} ($(RSPB.west|-RSPA)+(0mm,3mm)$);
            \node[below=0mm of m] (vdots) {\scalebox{.8}{$\vdots$}};
            \draw[<->] ($(RSPA.east)+(0mm,-3mm)$) -- node[midway] (m) {} node[midway, above, inner sep=1mm]{} ($(RSPB.west|-RSPA)+(0mm,-3mm)$);
          }
        \end{bbrbox}
      \end{bbroracle}
      \bbroracleqryspace{28mm}
      \bbroracleqryto{top=$\delta$}
      \bbroracleqryfrom{top=$s$}
    \end{bbrenv}
    \caption{\label{fig:honestubqc} $\sf{UBQC}$ with one qubit when both Alice and Bob follows the protocol honestly (see Protocol~\ref{protocol:ubqcReal})}
  \end{figure}
\end{theorem}
\begin{proof}
  In order to prove this theorem, we will proceed by contradiction. Let us assume that there exists $(P_A,P_B)$, and a simulator $\sigma$ having the above properties. \\
  Then, for the same resource $\cS_{UBQC1}$ we consider a different protocol $\protocol' = (\pi_A',\pi_B')$ that realizes it, but using a different filter\footnote{\label{footnote:filterInside} Note that we could include this new filter inside $\cS_{UBQC1}$ and use a more traditional filter $\filter^{c=0}$ but for simplicity we will just use a different filter.} $\filter^\sigma$ and a different simulator $\sigma'$:
    \begin{align}
    \pi_A' \cC \pi_B' &\approx_\eps \cS_{UBQC1} \filter^\sigma \label{eq:correctnessrsp2}\\
    \pi_A' \cC &\approx_\eps \cS_{UBQC1} \sigma' \label{eq:soundnessrsp2}
  \end{align}
  More specifically, the new filter $\filter^\sigma_{UBQC1}$ will depend on $\sigma$ defined in \cref{eq:soundnessrsp1}. Then our main proof can be described in the following steps:
  \begin{enumerate}
  \item We first show in \cref{lemma:SBQC1CCrealizable} that $\cS_{UBQC1}$ is also $\eps$-classically-realizable by $(\pi_A', \pi_B')$ with the filter $\filter^\sigma$.
  \item We then prove in \cref{lemma:SBQC1CC_RSP} that the resource $\cS_{UBQC1}$ is an $\sf{RSP}$ within $\negl[n]$, with respect to some well chosen converters $\cA$ and $\cQ$ (see \cref{fig:APiAPiBQ}) and this new filter $\filter^\sigma$.
  \item Then, we use the main result about $\sf{RSP}$ (\cref{thm:nogoClassicalRSP}) to show that $\cS_{UBQC1}$ is describable within $\negl[n]$ with respect to $\cA$ (\cref{corollarry:sbqc1_describable}).
  \item Finally, in \cref{lemma:impossibleDescribable} we prove that if $\cS_{UBQC1}$ is describable then we could achieve \emph{superluminal signaling}, which concludes the contradiction proof.
  \end{enumerate}
\end{proof}

\begin{definition}\label{def:SBQC1CC}
  Let $\pi' = (\pi'_A, \pi'_B)$ the protocol realizing $\cS_{UBQC1}$ described in the following way (as pictured \cref{fig:APiAPiBQ}):
  \begin{itemize}
      \item $\pi'_A = \pi_A$ (\cref{fig:honestubqc})
      \item $\pi'_B$: runs $P_B$, obtains a state $\rho_B$, then uses the angle $\delta$ received from its inner interface to compute $\tilde{\rho} := R_Z(-\delta)\rho_B$, and finally outputs $\tilde{\rho}$ on its outer interface and $s := 0$ on its inner interface.
  \end{itemize}
  Then we define $\filter^\sigma = \sigma \pi_B'$ depicted in \cref{fig:SBQC1CC} (with $\sigma$ the simulator defined in \cref{eq:soundnessrsp1} as explained before). \\
  We define the converters $\cA$ and $\cQ$ as seen in:
  \begin{figure}
    \centering
\includegraphics[width=.8\linewidth]{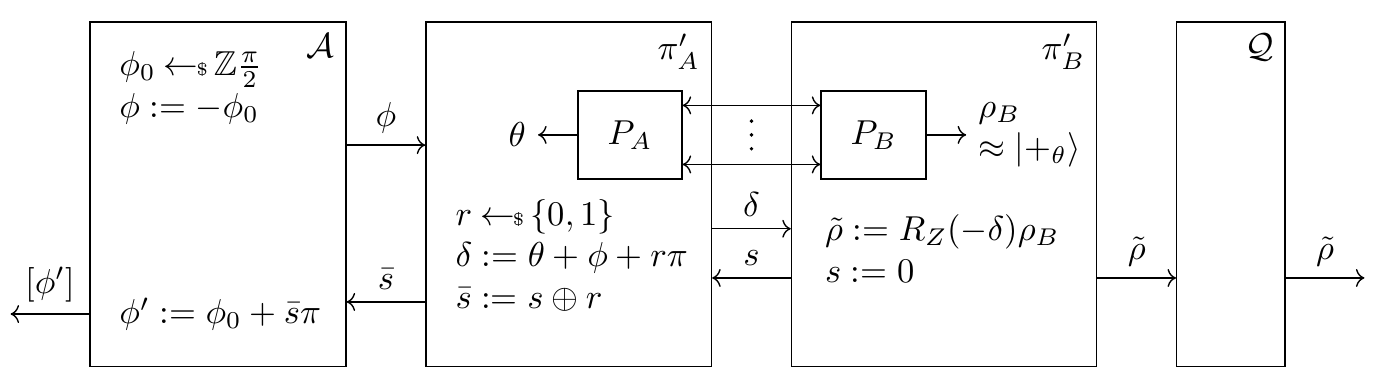}
    \caption{Definition of $\cA$, $\pi_A'$, $\pi_B'$ and $\cQ$.}
    \label{fig:APiAPiBQ}
  \end{figure}
  \begin{figure}[ht]
    \centering
    \includegraphics{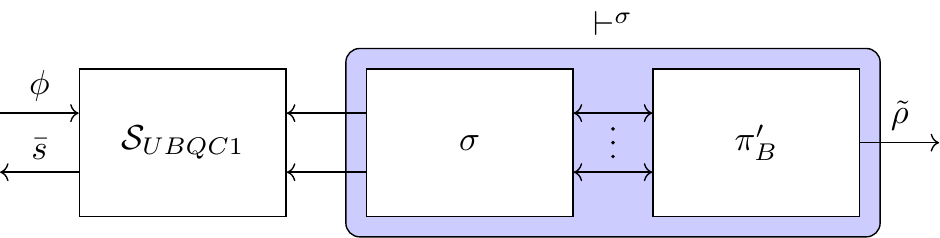}
    \caption{Description of $\filter^\sigma$}
    \label{fig:SBQC1CC}
  \end{figure}
  \end{definition}

\begin{lemma}\label{lemma:SBQC1CCrealizable}
  If $\cS_{UBQC1}$ is $\eps$-classically-realizable by $(\pi_A, \pi_B)$ with the filter $\filter^{c=0}$ then $\cS_{UBQC1}$ is also $\eps$-classically-realizable by $(\pi_A', \pi_B')$ with the filter $\filter^\sigma$.
\end{lemma}

\begin{proof}
  If $S_{UBQC1}$ is $\eps$-classically-realizable with $\filter^{c=0}$, then as seen in \cref{thm:nogo_ubqc1}, we have:
   \begin{align}
    \pi_A \cC \pi_B &\approx_\eps \cS_{UBQC1} \filter^{c=0} \label{eq:correctnessrsp2a}\\
    \pi_A \cC &\approx_\eps \cS_{UBQC1} \sigma \label{eq:soundnessrsp2a}
  \end{align}
  Now we can show that $\cS_{UBQC1}$ is $\eps$-classically-realizable by $(\pi_A',\pi_B')$ with $\filter^\sigma$ , i.e. that there exists a simulator $\sigma'$ such that:
   \begin{align}
    \pi_A' \cC \pi_B' &\approx_\eps \cS_{UBQC1} \filter^\sigma \label{eq:correctnessrsp3}\\
    \pi_A' \cC &\approx_\eps \cS_{UBQC1} \sigma' \label{eq:soundnessrsp3}
  \end{align}
For the correctness condition, we have:
\begin{align}
    \pi_A' \cC \pi_B' &= (\pi_A \cC) \pi_B' \\
    & \approx_\eps (\cS_{UBQC1} \sigma) \pi_B' \\
    & = \cS_{UBQC1} \filter^\sigma
\end{align}
For the security condition, we define $\sigma' = \sigma$. Then, we have:
\begin{align}
    \pi_A' \cC &= \pi_A \cC \\
    &\approx_\eps \cS_{UBQC1} \sigma
\end{align}
Which concludes our proof.
\end{proof}

\begin{lemma}\label{lemma:SBQC1CC_RSP}
 If $\cS_{UBQC1}$ is $\negl$-classically-realizable with $\filter^{c=0}$ then $\cS_{UBQC1}$ is an \emph{$\negl$-remote state preparation resource} with respect the converters $\cA$ and $\cQ$ and filter $\filter^\sigma$ defined in \cref{fig:APiAPiBQ}.
\end{lemma}

\begin{proof}
 We need to prove that:
 \begin{equation}
    \esp[([\rho], \rho_B) \leftarrow \cA \cS_{UBQC1} \filter^\sigma \cQ]{ \Tr(\rho \rho_B)} \geq 1-\eps
    \label{eq:remoteresourcea}
  \end{equation}
  First, we remark that due to \cref{lemma:SBQC1CCrealizable}:
  \begin{align}
      \cA \cS_{UBQC1} \filter^\sigma \cQ &\approx_\eps \cA \pi_A' \cC \pi_B' \cQ
      \label{eq:qSBQCapproxpibprime}
  \end{align}
  However, from the protocol description it is easy to check that in the real world $\bar{s} = 0 \xor r = r$, and therefore $\phi' := \phi_0 + \bar{s}\pi = \phi_0 + r\pi$ and $\rho = \ketbra{+_{\phi'}}$. And because the trace distance between $\rho_B$ and $\ketbra{+_\theta}$ is negligible with overwhelming probability (by the correctness of $(P_A,P_B)$), then we also have that $\tilde{\rho} = R_Z(-\delta)\rho_B R(-\delta)^\dagger$ is negligibly close in trace distance to $\ketbra{+_{\theta - \delta}} = \ketbra{+_{-\phi_0 + r\pi}} = \ketbra{+_{\phi'}}$. Therefore, we have:
  \begin{equation}
    \esp[([\rho], \tilde{\rho}) \leftarrow \cA \pi_A' \cC \pi_B' \cQ]{ \Tr(\rho \tilde{\rho})} \geq 1 - \negl[n]
  \end{equation}
  Then it also means that:
  \begin{equation}
    \esp[([\rho], \tilde{\rho}) \leftarrow \cA \cS_{UBQC1} \filter^\sigma \cQ]{ \Tr(\rho \tilde{\rho})} \geq 1 - \negl[n]
  \end{equation}
  otherwise we could (using a similar argument to the one given in the proof of \cref{thm:nogoClassicalRSP}) distinguish between the ideal and the real world, contradicting \cref{eq:qSBQCapproxpibprime}, which concludes the proof.
\end{proof}

Now, using our main \cref{thm:nogoClassicalRSP} we obtain directly that if $\cS_{UBQC1}$ is classically-realizable and $\sf{RSP}$ with respect to filter $\filter^\sigma$, then it is also describable:
\begin{corollary}
If $\cS_{UBQC1}$ is $\negl$-classically-realizable with respect to filter $\filter^{c=0}$ then $\cS_{UBQC1}$ is $\negl$-describable with respect to the converter $\cA$ described above.
\label{corollarry:sbqc1_describable}
\end{corollary}

\begin{lemma}\label{lemma:round}
  Let $\Omega = \{[\rho_i]\}$ be a set of (classical descriptions of) density matrices, such that $\forall i \neq j$, $\Tr(\rho_i\rho_j) \leq 1 - \eta$. Then let $([\rho], [\tilde{\rho}])$ be two random variables (representing classical description of density matrices), such that $[\rho] \in \Omega$ and $\esp[([\rho], [\tilde{\rho}])]{\Tr(\rho\tilde{\rho})} \geq 1 - \eps$, with $\eta > 6\sqrt{\eps}$. Then, if we define the following ``rounding'' operation that rounds $\tilde{\rho}$ to the closest $\tilde{\rho}_r \in \Omega$:
  \begin{align}
    [\tilde{\rho}_r] := \Round_\Omega([\tilde{\rho}]) := \argmax_{[\tilde{\rho}_r] \in \Omega} \Tr(\tilde{\rho}_r\tilde{\rho})
  \end{align}
  Then we have:
  \begin{equation}
  \pr[([\rho], [\tilde{\rho}])]{\Round_\Omega([\tilde{\rho}]) = [\rho]} \geq 1-\sqrt{\eps}
  \end{equation}
  In particular, if $\eps = \negl[n]$, and $\eta \neq 0$ is a constant, $\pr{\Round_\Omega([\tilde{\rho}]) = [\rho]} \geq 1-\negl[n]$.
\end{lemma}

\begin{proof}
We know that $\esp[([\rho], [\tilde{\rho}])]{ \Tr(\rho \tilde{\rho})} \geq 1-\eps$. Therefore, using Markov inequality we get that:
\begin{align}
  \pr[([\rho], [\tilde{\rho}])]{1 - \Tr(\rho \tilde{\rho}) \geq \sqrt{\eps}} &\leq \frac{\esp{1 - \Tr(\rho \tilde{\rho})}}{\eps} \\
  \pr[([\rho], [\tilde{\rho}])]{\Tr(\rho \tilde{\rho}) \leq 1 - \sqrt{\eps}} &\leq \frac{\eps}{\sqrt{\eps}} \\
  \pr[([\rho], [\tilde{\rho}])]{\Tr(\rho \tilde{\rho}) \geq 1 - \sqrt{\eps}} &\geq 1 - \sqrt{\eps} \label{eq:trBiggerEps1}
\end{align}
But when $\Tr(\rho \tilde{\rho}) \geq 1 - \sqrt{\eps}$, we have $\Round_\Omega([\tilde{\rho}]) = \rho$.
\begin{subproof}
  We will indeed show that $\forall \rho_i \in \Omega$, $\Tr(\rho_i\tilde{\rho}) \leq \Tr(\rho\tilde{\rho})$. By contradiction, we assume there exists $\rho_i \in \Omega$ such that $\rho_i \neq \rho$ and $\Tr(\rho_i\tilde{\rho}) > \Tr(\rho\tilde{\rho}) \geq 1 - \sqrt{\eps}$. But due to \cref{lemma:trace_transitivity} we have:
  \begin{align}
  \Tr(\rho_i\rho) \geq 1 - 3(\sqrt{\eps} + \sqrt{\eps}) = 1-6\sqrt{\eps}  
  \end{align}
  However, because both $\rho_i$ and $\rho$ belong to $\Omega$, we also have $\Tr(\rho_i\rho) \leq 1 - \eta < 1-6\sqrt{\eps}$, which is absurd.
\end{subproof}
Therefore, using \cref{eq:trBiggerEps1} we have
\begin{align}
  \pr[([\rho], [\tilde{\rho}])]{\Round_\Omega([\tilde{\rho}]) = [\rho]} \geq 1 - \sqrt{\eps}
\end{align}
which concludes the proof.
\end{proof}

\begin{lemma}\label{lemma:impossibleDescribable}
  $\cS_{UBQC1}$ cannot be $\negl$-describable with respect to  converter $\cA$.
\end{lemma}

\begin{proof}
 \begin{figure}
    \centering
\includegraphics[width=.8\linewidth]{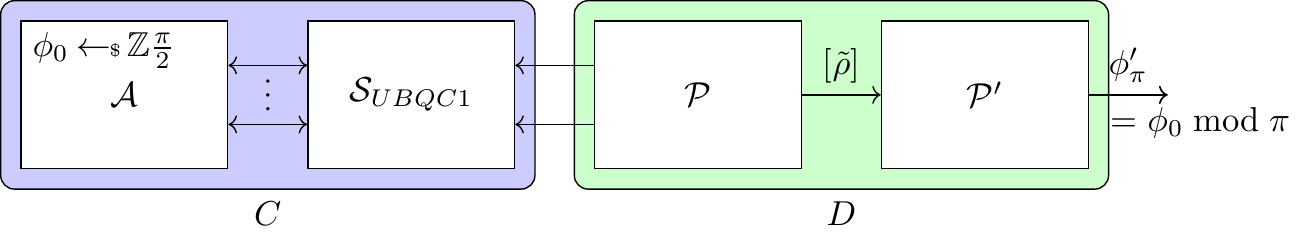}
    \caption{Illustration of the no-signaling argument}
    \label{fig:signal}
\end{figure}

If we assume that $\cS_{UBQC1}$ is $\negl[n]$-describable, then there exists a converter $\cP$ (outputting $[\tilde{\rho}]$) such that:
 \begin{equation}
    \esp[([\rho], [\tilde{\rho}]) \leftarrow \cA \cS_{UBQC1} \cP]{ \Tr(\rho \tilde{\rho})} \geq 1-\negl[n]
    \label{eq:neglDescribable}
 \end{equation}
We define the set $\Omega:= \{ [\ketbra{+_{\theta'}}{+_{\theta'}}] \ | \ \theta' \in \{0, \pi/4,...,7\pi/4 \} \}$. For simplicity, we will denote in the following $[\theta] = [\ketbra{+_{\theta}}]$.\\

In the remaining of the proof, we are going to use the converters $\cA$ and $\cP$ together with the ideal resource $\cS_{UBQC1}$, to construct a 2-party setting that would achieve signaling, which would end our contradiction proof. More specifically, we will define a converter $D$ running on the right interface of $\cS_{UBQC1}$ which will manage to recover the $\phi_0$ chosen randomly by $\cA$. \\  
As shown in \cref{fig:signal}, if we define $C$ as $C := \cA \cS_{UBQC1}$ and $D$ the converter described above, then the setting can be seen equivalently as: $C$ chooses as random $\phi_0$ and $D$ needs to output $\phi_0 \bmod \pi$. 
This is however impossible, as no message is sent from $\cS_{UBQC1}$ to its right interface (as seen in \cref{fig:signal}) (and thus no message from $C$ to $D$), and therefore guessing $\phi_0$ is forbidden by the no-signaling principle \cite{GRW80}.

 We define $\cP'$ as the converter that, given $[\tilde{\rho}]$ from the outer interface of $\cP$ computes $[\tilde{\phi}] = \Round_\Omega([\tilde{\rho}])$ and outputs $\tilde{\phi}_{\pi} = \tilde{\phi} \bmod \pi$ (as depicted in \cref{fig:signal}). We will now prove that $\tilde{\phi}_\pi = \phi_0 \bmod \pi$ with overwhelming probability.\\

All elements in $\Omega$ are different pure states, and in finite number, so there exist a constant $\eta > 0$ respecting the first condition of \cref{lemma:round}. Moreover from  \cref{eq:neglDescribable} we have that $\cS_{UBQC1}$ is $\eps$-describable with $\eps = \negl[n]$, so we also have (for large enough $n$), $\eta > 6\sqrt{\eps}$. Therefore, from \cref{lemma:round}, we have that:
 \begin{equation}\label{eq:rounds}
  \pr[([\rho], [\tilde{\rho}])  \leftarrow \cA \cS_{UBQC1} \cP]{\Round_\Omega([\tilde{\rho}]) = [\rho]} \geq 1-\negl
  \end{equation}
  But using the definition of converter $\cA$, we have: $[\rho] = [\phi']$, where $\phi' = \phi_0 + \bar{s} \pi$, and hence $\phi' \bmod \pi = \phi_0 \bmod \pi$. Then, using the definition of $\cP'$, the \cref{eq:rounds} is equivalent to:
  \begin{equation}\label{eq:round_eq}
   \pr[([\phi'], \tilde{\phi}_\pi)  \leftarrow \cA \cS_{UBQC1} \cP \cP']{\tilde{\phi}_\pi = \phi_0 \bmod \pi} \geq 1 - \negl
  \end{equation}

  However, as pictured in \cref{fig:signal}, this can be seen as a game between $C = \cA\cS_{UBQC1}$ and $D = \cP \cP'$, where, as explained before, $C$ picks a $\phi_{0} \in \Z\frac{\pi}{2}$ randomly, and $D$ needs to output $\phi_0 \bmod \pi$. From \cref{eq:round_eq} $D$ wins with overwhelming probability, however, we know that since there is no information transfer from C to D, the probability of winning this game better than 1/4 (guessing both the bits at random) would imply signalling.
  
\end{proof}

\begin{remark}
The guessing game described at the end of the preceding proof can be generalized to the case when some (partial) information transfer from $C$ to $D$ takes place.
More precisely, whenever we consider a new resource together with some converters $\cA$ and $\cQ$, it is enough to show that this resource is not describable to prove that it is impossible to classically realize. To that purpose, it may as above be practical to define a guessing game similar to the above one, but without the nice property that no information flows from $C$ to $D$. Here, the connections with the non-local games~\cite{brunner2014bell} and information causality~\cite{pawlowski2009information} could provide an upper bound on the winning probability (e.g., as a function of the conditional mutual information conditioned on the information exchanged). 
We leave the quantitative analysis for future work.
\end{remark}

\subsection{Impossibility of Composable \texorpdfstring{$\sf{UBQC}_{CC}$}{UBQCcc} on Any Number of Qubits \label{subsec:imposs_general_cc_ubqc}}

We saw in \cref{thm:nogo_ubqc1} that it is not possible to implement a composable classical-client $\sf{UBQC}$ protocol performing a computation on a single qubit.
In this section, we prove that this result generalizes to the impossibility of $\sf{UBQC_{CC}}$ on computations using an arbitrary number of qubits. The proof works by reducing the general case to the single-qubit case from the previous section.

\begin{theorem}[{No-go Composable Classical-Client $\sf{UBQC}$}]\label{thm:nogo_ubqc}
  Let $(P_A, P_B)$ be a protocol interacting only through a classical channel $\cC$, such that $(\theta, \rho_B) \leftarrow (P_A \cC P_B)$ with $\theta \in \Z \frac{\pi}{4}$, and such that the trace distance between $\rho_B$ and $\ket{+_\theta}\bra{+_\theta}$ is negligible with overwhelming probability. Then, if we define $(\pi^G_A,\pi^G_B)$ as the $\sf{UBQC}$ protocol on any fixed graph $G$ (with at least one output qubit\footnote{Note, that in $\sf{UBQC_{CC}}$ with zero output qubits the client does not receive any results. Hence, the protocol is trivially implementable for this degenerated case.}), that uses $(P_A, P_B)$ as a sub-protocol to replace the quantum channel, $(\pi^G_A, \pi^G_B)$ is not composable, i.e. there exists no simulator $\sigma$ such that:
  \begin{align}
    \pi^G_A \cC \pi^G_B &\approx_\eps \cS_{UBQC1} \filter^{c=0} \label{eq:correctnessubqc1}\\
    \pi^G_A \cC &\approx_\eps \cS_{UBQC1} \sigma \label{eq:soundnessubqc1}
  \end{align}
  for some negligible $\eps = \negl[n]$.
\end{theorem}
\begin{proof}
  To prove this statement, we just need to prove that we can come back to the setting with a single qubit, where we want to perform a computation with angle $\phi$, and output one angle close to $\phi$  as in the proof of \cref{thm:nogo_ubqc1}. Because the graph has at least one output qubit, we will denote by $\omega$ the index of the last output qubit. So the idea is to let the distinguisher choose the client input such that for any node $i \neq \omega$ in the graph, $\phi_i = 0$, and for the output qubit, $\phi_\omega = \phi$. Moreover, on the server side, the distinguisher will behave like the honest protocol $\pi^G_B$, except that it will not entangle the qubits provided by $P_A$, and it will deviate on the output qubit $\omega$ by not measuring it and sending $s := 0$, the qubit being rotated again with angle $-\delta_\omega$, and outputed on the outer interface, like in the one-qubit case. It is now easy to see by induction (over the index of the qubit, following the order chosen on $G$) that, in the real world, for all $i \neq \omega$, we always have $s_i = r_i$, therefore $\bar{s_i} = 0$. So for all nodes $i$, (including $\omega$), $s^X_i = \oplus_{i \in D_{i}} \bar{s}_i = 0$ and $s^Z_i = \oplus_{i \in D'_{i}} \bar{s}_i = 0$. Thus we have on the last node:
  \begin{align*}
    \delta_\omega
    &= \theta_\omega + (-1)^{s^X_{\omega}}\phi_\omega + s^Z_{\omega}\pi+ r_\omega \pi\\
    &= \theta_\omega + \phi + r_\omega \pi
  \end{align*}
  which corresponds exactly to the single-qubit setting, shown to be impossible.
\end{proof}


\section{Game-Based Security of \texorpdfstring{$\sf{QF}$-$\sf{UBQC}$}{QF-UBQC} \label{sec:game_based_cubqc}}

While we know from Theorem~\ref{thm:nogo_ubqc} that classical-client $\sf{UBQC}$ (henceforth simply $\sf{UBQC_{CC}}$) cannot be proven secure in a fully composable setting, there is hope that it remains possible with a weaker definition of security. And indeed, in this section we show that $\sf{UBQC_{CC}}$ is possible in the \emph{game-based setting} by implementing it using a combination of the known quantum-client $\sf{UBQC}$ Protocol~\ref{protocol:ubqcReal} \cite{broadbent2009universal} and 8-states QFactory Protocol~\ref{protocol:QFactory4to8} \cite{cojocaru2019qfactory}. We start with giving a formal definition of the game-based security of $\sf{UBQC_{CC}}$.

\begin{definition}[Blindness of $\sf{UBQC_{CC}}$] \label{def:game_based_security_blindness_CUBQC}
  A $\sf{UBQC_{CC}}$ protocol $\mathcal{P} = (P_C, P_S)$ is said to be (computationally) adaptively blind if no computationally bounded malicious server can distinguish between runs of the protocol with adversarially chosen measurement patterns on the same MBQC graph.
  
  In formal terms, $\mathcal{P}$ is said to be (computationally) adaptively blind if and only if for any quantum-polynomial-time adversary $A$ it holds that
    \begin{align*}
      \operatorname{Pr} \left[ c' = c \, \left| \, (\phi^{(1)}, \phi^{(2)}) \gets A ,\, c \sample \bin ,\, \left\langle P_C(\phi^{(c)}), A \right\rangle ,\, c' \gets A \right. \right] \leq \frac{1}{2} + \operatorname{negl}(\lambda),
  \end{align*}
  where $\lambda$ is the security parameter, and $\left\langle P_C(\phi^{(c)}), A \right\rangle$ denotes the interaction of the two algorithms $P_C(\phi^{(c)})$ and $A$.
\end{definition}

\begin{remark}
Although, \cref{def:game_based_security_blindness_CUBQC} is written using the terminology of measurement-based model. It doesn't compromise the generality, as the model is universal and can be easily translated into a circuit model, because the measurement pattern and unitary operator have a one-to-one mapping.
\end{remark}

\subsection{Implementing Classical-Client \texorpdfstring{$\sf{UBQC}$}{UBQC} with QFactory \label{cubqc_qfact}}

The $\sf{UBQC}$ protocol from~\cite{broadbent2009universal}, where the quantum interaction is replaced by a $\sf{RSP}^{8-states}_{CC}$ protocol, is shown in Protocol~\ref{protocol:ubqcReal}. In this section, we replace the $\sf{RSP}^{8-states}_{CC}$ protocol with the concrete protocol proposed in~\cite{cojocaru2019qfactory}.  This protocol, known by the name of 8-states QFactory\footnote{We refer here to the 8-states QFactory implementation with negligible abort probability, and superpolynomial parameters. This is necessary since our proof does not take the abort case into account for now.} and described in Protocol~\ref{protocol:QFactory4to8}, exactly emulates the capability of $\sf{RSP}^{8-states}_{CC}$. The resulting protocol contains a QFactory instance for each qubit that would have been generated on the client's side.
The keys to all QFactory instances are generated entirely independently by the client.

Unfortunately, considering the results from Section~\ref{sec:impossibility_composable_CUBQC} there is no hope that the composable security of any $\sf{UBQC_{CC}}$ may be achieved. Nonetheless, letting go of composability, we are able to prove the game-based security for this specific combination of protocols. This leads us to the main theorem of this section.

\begin{theorem}[Game-based Blindness of $\sf{QF}$-$\sf{UBQC}$] \label{thm:game_based_security_UBQC_QF}
    The protocol resulting from combining the quantum-client $\sf{UBQC}$ protocol with QFactory is a (computationally) adaptively blind implementation of $\sf{UBQC_{CC}}$ in the game-based model according to \cref{def:game_based_security_blindness_CUBQC}. We call this protocol $\sf{QF}$-$\sf{UBQC}$.
\end{theorem}

The proof of \cref{thm:game_based_security_UBQC_QF} which will be given in the remainder of this section follows two main ideas:
\begin{enumerate}
    \item Every angle used in the $\sf{UBQC}$ protocol has only eight possible values, and can, therefore, be described by three bits. In the protocol, the first bit is the one for which QFactory \emph{cannot} guarantee blindness. Fortunately, the additional one-time padding in $\sf{UBQC}$ allows analyzing the blindness of the protocol independently of the blindness of exactly this first bit. Therefore, it suffices to rely on the blindness of the last two bits which is conveniently guaranteed by QFactory and the hardness of LWE.
    \item To analyze the leakage about the last two bits during a QFactory run, it is sufficient to notice that the leakage is equal to a ciphertext under an LWE-based encryption scheme. The semantic security of this encryption scheme and the hardness assumption for LWE guarantee that this leakage is negligible and can be omitted.
\end{enumerate}

In more detail, the 8-states QFactory protocol which is used here consists of two combined runs of 4-states QFactory, each contributing with a single blind bit to the three-bit angles used in the $\sf{UBQC}$ protocol. Recall from \cref{thm:8_statescorrectness} and \cref{thm:security_qfac_2_1} the formulae for how these angles from the 4-states protocol are combined in the 8-states protocol. If $B_1$ is the hidden bit of the first 4-states QFactory instance and $B_1'$ the hidden bit of the second instance, then we obtain
\begin{align}
    L_1 = B_2' \oplus B_2 \oplus [B_1 \cdot (s_1 \oplus s_2)], \; L_2 = B_1' \oplus [(B_2 \oplus s_2) \cdot B_1], \; L_3 = B_1,
\end{align}
where $L = L_1L_2L_3 \in \bin^3$ is the description of the output state $\left| +_{L\frac{\pi}{4}} \right\rangle$, $s_1, s_2$ are computed by the server, and
\begin{align}
    B_2 = f(\sk, B_1, y, b), \qquad B_2' = f(\sk', B_1', y', b')
\end{align}
for some function $f$, QFactory secret keys $\sk, \sk'$, and server-chosen values $y, b, y', b'$.

The two 4-states QFactory instances now leak the ciphertext of $B_1$ and $B_1'$, respectively. Given the semantic security of the encryption, after a run of 8-states QFactory, $L_2$ and $L_3$ remain hidden, while the blindness of $L_1$ cannot be guaranteed by QFactory. This fact is going to be useful in the following proof.

\subsection{Single-Qubit \texorpdfstring{$\sf{QF}$-$\sf{UBQC}$}{QF-UBQC} \label{subsection:game_single_qubit_case}}

We first prove the security of combining QFactory with $\sf{UBQC}$ on a single qubit.

\begin{lemma}[Blindness in the single-qubit case] \label{lemma:single_qubit_case}
    The protocol resulting from combining the quantum-client $\sf{UBQC}$ protocol with (8-states) QFactory is a (computationally) adaptively blind implementation of $\sf{UBQC_{CC}}$ in the game-based model for MBQC computations on a single qubit.
\end{lemma}

\begin{proof}
We start with the real protocol, describing the adaptive blindness of QFactory combined with single-qubit $\sf{UBQC}$.    In the following, we denote the set of possible angles by $M = \{ j\pi/4, j=0,\dots,7 \}$.
The encryption scheme that appears in Game~1 is the semantically secure public-key encryption scheme from \cite{regev2009lattices}. Note that the two key pairs are generated completely independently on the challenger's side.

\noindent \begin{minipage}{\textwidth}
    \null\hfill \textsc{Game 1}: \\[-3pt]
    \fbox{
    
    \pseudocode[codesize=\scriptsize, width=\textwidth, colsep=-0em, addtolength=0em]{
        \qquad\;\;\; \textbf{Adversary} \> \> \textbf{Challenger} \\[0.1\baselineskip][\hline]
        \> \> \\[-0.6\baselineskip]
        \pcln \hspace{-5pt}\text{Choose } \phi^{(1)}, \phi^{(2)} \in M \> \sendmessageright*[3cm]{\phi^{(1)}, \phi^{(2)}} \>  \hspace{-5pt}c \sample \bin \\[-5pt]
        \pcln \> \> \hspace{-5pt}B_{1}, B_{1}' \sample \bin \\[-5pt]
        \pcln \> \sendmessageleft*[3cm]{\begin{matrix} \pk, \pk', \\ \operatorname{Enc}^{\pk}(B_{1}), \operatorname{Enc}^{\pk'}(B_{1}') \end{matrix}} \> \hspace{-5pt}\text{Generate key pairs } (\sk, \pk), (\sk', \pk') \\[-8pt]
        \pcln \> \sendmessageright*[3cm]{y, b, y', b', s_1, s_2} \> \hspace{-5pt}B_2 = f(\sk, B_1, y, b),\, B_2' = f(\sk', B_1', y', b') \\[-8pt]
        \pcln \> \> \hspace{-5pt}L_1 = B_2' \oplus B_2 \oplus [B_1 \cdot (s_1 \oplus s_2)] \\
        \pcln \> \> \hspace{-5pt}L_2 = B_1' \oplus [(B_2 \oplus s_2) \cdot B_1] \\
        \pcln \> \> \hspace{-5pt}L_3 = B_1 \\
        \pcln \> \> \hspace{-5pt}r \sample \bin \\
        \pcln \> \sendmessageleft*[3cm]{\delta} \> \hspace{-5pt}\delta = \phi^{(c)} + L_{3} \pi /4 + L_{2} \pi /2 + L_{1} \pi + r \pi \\[-8pt]
        \pcln \> \sendmessageright*[3cm]{s} \> \\[-8pt]
        \qquad\;\;\; \hspace{-5pt}\text{Compute guess } \\[-10pt]
        \pcln \hspace{-5pt}c' \in \bin \> \sendmessageright*[3cm]{c'} \> \hspace{-5pt}\text{Check } c' = c ?
    }
    
}
    \vspace{0.2cm}
\end{minipage}

In the following, instead of repeating the redundant parts of subsequent games, we only present incremental modifications to Game~1. Every not explicitly written line is assumed to be identical to the previous game.

Clearly, since $s$ is never used by the challenger, we can remove it from the protocol without distorting the success probability of the adversary. Next, we remove $L_1$ from the protocol and from the calculation of $\delta$. $L_1$ is only used in the calculation of $\delta$, which can be rewritten as
\begin{align}
    \delta = \phi^{(c)} + L_{3} \pi /4 + L_{2} \pi /2 + (L_{1} + r) \pi.
\end{align}
Since $r$ is a uniform binary random variable with unique use in this line, $(L_1 + r)$ is still uniform over $\bin$. Therefore, removing $L_1$ leaves the distribution of the protocol outcome unchanged.

\noindent \begin{minipage}{\textwidth}
    \null\hfill \textsc{Game 2}: \\[-3pt]
    \fbox{\pseudocode[codesize=\scriptsize, lnstart=3, width=\textwidth]{
        \;\;\vdots \\[-8pt]
        \pcln \> \sendmessageright*[4cm]{y, b, y', b', s_1, s_2} \> B_2 = f(\sk, B_1, y, b),\, \tikzmark{start0} B_2' = f(\sk', B_1', y', b') \tikzmark{stop0} \\[-8pt]
        \tikzmark{start1a} \pcln \tikzmark{stop1a}  \> \> \tikzmark{start1b}  L_1 = B_2' \oplus B_2 \oplus [B_1 \cdot (s_1 \oplus s_2)]  \tikzmark{stop1b} \\[-8pt]
        \;\;\vdots \\[-8pt]
        \setcounter{pclinenumber}{8}
        \pcln \> \sendmessageleft*[4cm]{\delta} \> \delta = \phi^{(c)} + L_{3} \pi /4 + L_{2} \pi /2 \tikzmark{start2} + L_{1} \pi \tikzmark{stop2} + r \pi \\[-8pt]
        \tikzmark{start3a}  \pcln \tikzmark{stop3a} \> \sendmessageright*[4cm]{s} 
        \> \\[-12pt]
        \;\;\vdots
    }
    }

    \begin{tikzpicture}[remember picture, overlay]
        \draw[red,thick] ([yshift=0.5ex]pic cs:start0) -- ([yshift=0.5ex]pic cs:stop0);
        \draw[red,thick] ([yshift=0.5ex]pic cs:start1a) -- ([yshift=0.5ex]pic cs:stop1a);
        \draw[red,thick] ([yshift=0.5ex]pic cs:start1b) -- ([yshift=0.5ex]pic cs:stop1b);
        \draw[red,thick] ([yshift=0.5ex]pic cs:start2) -- ([yshift=0.5ex]pic cs:stop2);
        \draw[red,thick] ([yshift=0.5ex]pic cs:start3a) -- ([yshift=0.5ex]pic cs:stop3a);
        \draw[red,thick] (4.3,0.9) -- (4.7,0.9);
    \end{tikzpicture}
    \vspace{0.2cm}
\end{minipage}

The next step introduces a (negligible) distortion to the success probability of the adversary.
By the semantic security of the employed encryption scheme, no quantum-polynomial-time adversary can notice if the plaintext is replaced by pure randomness except with negligible probability, even if information about the original plaintext is leaked on the side. Therefore, replacing $B_1'$ in the encryption by independent randomness cannot lead to a significant change of the adversary's success probability.
Further, since ciphertexts of independent randomness can be equally generated by the adversary herself (being in possession of the public key), we can remove the encryption of $B_1'$ from the protocol altogether.

\noindent \begin{minipage}{\textwidth}
	\null\hfill \textsc{Game 3}: \\[-3pt]
	\fbox{\pseudocode[codesize=\scriptsize, width=\textwidth]{
		\;\;\vdots \\[-8pt]
		\setcounter{pclinenumber}{2}
		\pcln \> \sendmessageleft*[4cm]{\pk, \pk', \operatorname{Enc}^{\pk}(B_{1}), \operatorname{Enc}^{\pk'}(B_{1}')}
		 \> \text{Generate key pairs } (\sk, \pk), \tikzmark{start6} (\sk', \pk') \tikzmark{stop6} \\[-12pt]
		\;\;\vdots
	}}
	\begin{tikzpicture}[remember picture, overlay]
		\draw[red,thick] ([yshift=0.5ex]pic cs:start4) -- ([yshift=0.5ex]pic cs:stop4);
		\draw[red,thick] ([yshift=0.5ex]pic cs:start5) -- ([yshift=0.5ex]pic cs:stop5);
		\draw[red,thick] (5.2,.9) -- (6.6,.9);
		\draw[red,thick] (3.3,.9) -- (3.7,.9);
		\draw[red,thick] ([yshift=0.5ex]pic cs:start6) -- ([yshift=0.5ex]pic cs:stop6);
	\end{tikzpicture}
	\vspace{0.2cm}
\end{minipage}

Next, note that $B_1'$ perfectly one-time pads the value of $L_2$. This breaks the dependency of $L_2$ on $B_2$, $s_2$ and $B_1$. It does not change the distribution of $L_2$, if $L_2$ is instead directly sampled uniformly from $\bin$.
Since $B_2$ is unused, we remove it in the following game, and $y, b, y', b', s_1, s_2$ can be ignored.

\noindent \begin{minipage}{\textwidth}
	\null\hfill \textsc{Game 4}: \\[-3pt]
	\fbox{\pseudocode[codesize=\scriptsize, width=\textwidth]{
		\;\;\vdots \\[-8pt]
		\setcounter{pclinenumber}{3} 
		\tikzmark{start7a} \pcln \tikzmark{stop7a} \> \sendmessageright*[4cm]{y, b, y', b', s_1, s_2}
		\> \tikzmark{start7c} B_2 = f(\sk, B_1, y, b) \tikzmark{stop7c} \\[-12pt]
		\;\;\vdots \\[-3pt]
		\setcounter{pclinenumber}{5}
        \pcln \> \> \tikzmark{start8} L_2 = B_1' \oplus [(B_2 \oplus s_2) \cdot B_1] \tikzmark{stop8} \quad \textcolor{red}{L_2 \sample \bin} \\[-8pt]
        \;\;\vdots
	}}
	\begin{tikzpicture}[remember picture, overlay]
		\draw[red,thick] ([yshift=0.5ex]pic cs:start7a) -- ([yshift=0.5ex]pic cs:stop7a);
		\draw[red,thick] ([yshift=0.5ex]pic cs:start7b) -- ([yshift=0.5ex]pic cs:stop7b);
		\draw[red,thick] ([yshift=0.5ex]pic cs:start7c) -- ([yshift=0.5ex]pic cs:stop7c);
		  \draw[red,thick] (3.2,1.7) -- (5.6,1.7);
		\draw[red,thick] ([yshift=0.5ex]pic cs:start8) -- ([yshift=0.5ex]pic cs:stop8);
		
	\end{tikzpicture}
	\vspace{0.2cm}
\end{minipage}
By the same argument as for the transition from Game~2 to Game~3, we remove the encryption of $B_1$ from the following game. This introduces at most a negligible change in the success probability of the adversary.

Finally, since the encryption scheme is not in use anymore, we can also remove the key generation and the message containing the public key without affecting the adversary's success probability.

\noindent \begin{minipage}{\textwidth}
	\null\hfill \textsc{Game 5}: \\[-3pt]
	\fbox{\pseudocode[codesize=\scriptsize, width=\textwidth]{
		\;\;\vdots \\[-8pt]
		\setcounter{pclinenumber}{2}
		\tikzmark{start9a} \pcln \tikzmark{stop9a} \> \sendmessageleft*[4cm]{\pk, \operatorname{Enc}^{\pk}(B_{1})}
		 \> \tikzmark{start9c} \text{Generate key pair } (\sk, \pk) \tikzmark{stop9c} \\[-12pt]
		\;\;\vdots
	}}
	\begin{tikzpicture}[remember picture, overlay]
		\draw[red,thick] ([yshift=0.5ex]pic cs:start9a) -- ([yshift=0.5ex]pic cs:stop9a);
		  \draw[red,thick] (4.5,.9) -- (6.3,.9);
		\draw[red,thick] ([yshift=0.5ex]pic cs:start9b) -- ([yshift=0.5ex]pic cs:stop9b);
		\draw[red,thick] ([yshift=0.5ex]pic cs:start9c) -- ([yshift=0.5ex]pic cs:stop9c);
	\end{tikzpicture}
	\vspace{0.2cm}
\end{minipage}

We now see that $\delta$ is a uniformly random number, $L_2, L_3$, and $r$ being i.i.d. uniform bits. Therefore, the calculation and the message containing $\delta$ can be removed from the protocol without affecting the adversary.

\noindent \begin{minipage}{\textwidth}
	\null\hfill \textsc{Game 6}: \\[-3pt]
	\fbox{\pseudocode[codesize=\scriptsize, width=\textwidth]{
		\;\;\vdots \\
		\setcounter{pclinenumber}{1}
		\tikzmark{start10a} \pcln \tikzmark{stop10a} \> \> \tikzmark{start10b} B_{1}, B_{1}' \sample \bin \tikzmark{stop10b} \\[-8pt]
		\;\;\vdots \\
		\setcounter{pclinenumber}{5}
        \tikzmark{start11a} \pcln \tikzmark{stop11a} \> \> \tikzmark{start11b} L_2 \sample \bin \tikzmark{stop11b} \\
        \tikzmark{start12a} \pcln \tikzmark{stop12a} \> \> \tikzmark{start12b} L_3 = B_1 \tikzmark{stop12b} \\
		\tikzmark{start13a} \pcln \tikzmark{stop13a} \> \> \tikzmark{start13b} r \sample \bin \tikzmark{stop13b} \\[-5pt]
		\tikzmark{start14a} \pcln \tikzmark{stop14a} \> \sendmessageleft*[4cm]{\delta}
		 \> \tikzmark{start14c} \delta = \phi^{(c)} + L_{3} \pi /4 + L_{2} \pi /2 + r \pi \tikzmark{stop14c} \\[-12pt]
		\;\;\vdots
	}}
	\begin{tikzpicture}[remember picture, overlay]
		\draw[red,thick] ([yshift=0.5ex]pic cs:start10a) -- ([yshift=0.5ex]pic cs:stop10a);
		\draw[red,thick] ([yshift=0.5ex]pic cs:start10b) -- ([yshift=0.5ex]pic cs:stop10b);
		\draw[red,thick] ([yshift=0.5ex]pic cs:start11a) -- ([yshift=0.5ex]pic cs:stop11a);
		\draw[red,thick] ([yshift=0.5ex]pic cs:start11b) -- ([yshift=0.5ex]pic cs:stop11b);
		\draw[red,thick] ([yshift=0.5ex]pic cs:start12a) -- ([yshift=0.5ex]pic cs:stop12a);
		\draw[red,thick] ([yshift=0.5ex]pic cs:start12b) -- ([yshift=0.5ex]pic cs:stop12b);
		\draw[red,thick] ([yshift=0.5ex]pic cs:start13a) -- ([yshift=0.5ex]pic cs:stop13a);
		\draw[red,thick] ([yshift=0.5ex]pic cs:start13b) -- ([yshift=0.5ex]pic cs:stop13b);
		\draw[red,thick] ([yshift=0.5ex]pic cs:start14a) -- ([yshift=0.5ex]pic cs:stop14a);
		\draw[red,thick] ([yshift=0.5ex]pic cs:start14b) -- ([yshift=0.5ex]pic cs:stop14b);
		\draw[red,thick] ([yshift=0.5ex]pic cs:start14c) -- ([yshift=0.5ex]pic cs:stop14c);
		\draw[red,thick] (4.9,.95) -- (5.2,.95);
	\end{tikzpicture}
	\vspace{0.2cm}
\end{minipage}

In Game~6, the inputs of the adversary are ignored by the challenger. Therefore, the computation angles $\phi^{(1)}$, $\phi^{(2)}$ can equally be removed from the protocol, leaving us with the final Game~7.

\noindent \begin{minipage}{\textwidth}
	\null\hfill \textsc{Game 7}: \\[-3pt]
	\fbox{\pseudocode[codesize=\scriptsize, width=\textwidth]{
		\qquad\;\;\; \textbf{Adversary} \> \> \textbf{Challenger} \\[0.1\baselineskip][\hline]
		\> \> \\[-0.6\baselineskip]
		\setcounter{pclinenumber}{0} 
		\pcln \tikzmark{start15} \text{Choose } \phi^{(1)}, \phi^{(2)} \in M \tikzmark{stop15} \> \sendmessageright*[4cm]{\phi^{(1)}, \phi^{(2)} }
		\>  c \sample \bin \\[-5pt]
		\setcounter{pclinenumber}{10}
		\pcln \text{Compute guess } c' \in \bin \> \sendmessageright*[4cm]{c'} \> \text{Check } c' = c ?
	}}
	\begin{tikzpicture}[remember picture, overlay]
		\draw[red,thick] ([yshift=0.5ex]pic cs:start15) -- ([yshift=0.5ex]pic cs:stop15);
		\draw[red,thick] ([yshift=0.5ex]pic cs:start16) -- ([yshift=0.5ex]pic cs:stop16);
		\draw[red,thick] (7.2,1.25) -- (8.2,1.25);
	\end{tikzpicture}
	\vspace{0.2cm}
\end{minipage}

Game~7 exactly describes the adversary's uninformed guess of the outcome of an independent bit flip. Therefore, by a simple information-theoretic argument, any strategy for the adversary will lead to a success probability of exactly $1/2$.

We summarize:

\begin{align*}
    &\text{Succ-Pr}_\text{Game1} = \text{Succ-Pr}_\text{Game2}, \quad \left| \text{Succ-Pr}_\text{Game2} - \text{Succ-Pr}_\text{Game3} \right| \leq \operatorname{negl}(\lambda), \\
    &\text{Succ-Pr}_\text{Game3} = \text{Succ-Pr}_\text{Game4}, \quad \left| \text{Succ-Pr}_\text{Game4} - \text{Succ-Pr}_\text{Game5} \right| \leq \operatorname{negl}(\lambda), \\
    &\text{Succ-Pr}_\text{Game5} = \text{Succ-Pr}_\text{Game6} = \text{Succ-Pr}_\text{Game7} = \frac{1}{2},
\end{align*}
and therefore we have 
    $\left| \text{Succ-Pr}_\text{Game1} - \frac{1}{2} \right| \leq \operatorname{negl}(\lambda)$
concluding the proof.
\end{proof}

subsection{General \texorpdfstring{$\sf{QF}$-$\sf{UBQC}$}{QF-UBQC}  \label{subsec:game_based_general}}

We extend the security proof from \cref{subsection:game_single_qubit_case} to $\sf{UBQC}$ on polynomially-sized graphs, i.e. MBQC computations on a polynomial number of qubits. The proof works by induction over the number $n$ of qubits in the graph. \cref{lemma:single_qubit_case} with $n=1$ serves as start of the induction. We continue with proving the induction step, assuming the security of $\sf{QF}$-$\sf{UBQC}$ on graphs of size $n$ and showing its security for any graph of size $n+1$. The induction step works analogously to the proof of \cref{lemma:single_qubit_case}. In this way, the security of $\sf{QF}$-$\sf{UBQC}$ on $n$ qubits is reduced to the security of $\sf{QF}$-$\sf{UBQC}$ on $n-1$ qubits, which can be reduced to the security of $\sf{QF}$-$\sf{UBQC}$ on even one qubit less. This chain continues down to the single-qubit case whose security was already established in \cref{lemma:single_qubit_case}. Every step in this chain adds at most a negligible probability to the adversary's advantage. Therefore, also any such chain of polynomial length adds no more than a negligible probability to the adversary's advantage in the single-qubit case, thereby showing the security of the protocol on $n$ qubits. We now provide the full details of the induction step. 

\begin{proof}[Details of the proof of \cref{thm:game_based_security_UBQC_QF}]
The proof works by induction over the number $n$ of qubits in the graph. \cref{lemma:single_qubit_case} with $n=1$ serves as start of the induction. We continue with proving the induction step, assuming the security of $\sf{QF}$-$\sf{UBQC}$ on graphs of size $n$ and showing its security for any graph of size $n+1$.

We first state some useful observations for the proof:
\begin{enumerate}
    \item The existence of a \emph{flow} on the MBQC graph induces a total order of all qubits in the graph, the order in which the qubits are measured. We subsequently assume that in the protocol the qubits are processed in exactly this order.
    \item Given this order on the qubits, the dependence of the computation angles $\delta_i$ on outcomes of measurement of other qubits takes a specific form, they solely depend on previous (corrected) measurement outcomes $\{ \bar{s}_j, j<i \}$, i.e. outcomes of measurements of qubits smaller in the order induced by the flow. Since the exact form of this dependence does not matter for the following proof, we denote the update of the angles in the following general way:
    \begin{align*}
        \delta_i = &(-1)^{f_1(s_1, r_1, \dots, s_{i-1}, r_{i-1})} \phi_i + \theta_1 \pi /4 + \theta_2 \pi /2 + \theta_3 \pi \\
        &+ r_i\pi + f_2(s_1, r_1, \dots, s_{i-1}, r_{i-1})\pi,
    \end{align*}
    with (deterministic families of) functions $f_1$ and $f_2$.
    \item Given the previous observation, one can generalize the statement of the theorem to a family of protocols for any functions $f_1$ and $f_2$. For the remainder of the proof, we do hence not assume anything about these two functions, but simply take them as given. The actual statement of the theorem then follows as a special case, imposing that $f_1$ and $f_2$ describe the MBQC correction terms.
\end{enumerate}

Given these observations, the rest of the proof works analogously to the proof of \cref{lemma:single_qubit_case}, removing one-by-one the ciphertexts of the two basis bits $B_1, B_1'$ of the last QFactory instance, before removing the last measurement angle $\delta$ and reducing the protocol on $n+1$ qubits to the protocol on one qubit less.
\end{proof}

By the inductive nature of this proof, every qubit -- and hence every QFactory instance -- adds some negligible value to the success probability of the malicious adversary. This explains that the security only holds for polynomially-sized graphs. For an MBQC graph on a superpolynomial number of qubits, there are no guarantees anymore that these small errors don't add up to something constant. Having in mind that QFactory is trivially broken by exponential adversaries, it is clear that this is the best we can expect.


\section*{Acknowledgements.}
The authors thank Céline Chevalier, Omar Fawzi, Daniel Jost, and Luka Music for useful discussions. LC also thanks M.T. This work has been supported in part by grant FA9550-17-1-0055, by the European Union’s H2020 Programme under grant agreement number ERC-669891, and by the French ANR Project ANR-18-CE39-0015 (CryptiQ). EK acknowledges support from the EPSRC Verification of Quantum Technology grant (EP/N003829/1), the EPSRC Hub in Quantum Computing and Simulation (EP/T001062/1), and the UK Quantum Technology Hub: NQIT grant (EP/M013243/1). LC and DL gratefully acknowledge support from the French ANR project ANR-18-CE47-0010 (QUDATA). LC, EK, and DL acknowledge funding from the EU Flagship Quantum Internet Alliance (QIA) project. AM gratefully acknowledges funding from the AFOSR MURI project ``Scalable Certification of Quantum Computing Devices and Networks''. This work was partly done while AM was at University of Edinburgh, UK where it was supported by EPSRC Verification of Quantum Technology grant (EP/N003829/1).

\bibliographystyle{alpha}
\bibliography{ref}


\appendix

\section{Game-Based Security and Constructive Cryptography \label{app:game_and_ac}}

The main aim of our work is to prove possibility and impossibility results in different security models. We will in this paper focus mostly on two different notions: the game-based security model and the Constructive Cryptography framework.

The definition of game-based security is pretty straightforward: we define a \emph{game} between a challenger and an (arbitrary) adversary: a protocol is secure if no adversary can win this game with ``good'' probability. The problem of this approach is that one game describes only one possible attack, and it is hard to list all the possible attacks against a protocol. Therefore, a protocol that proves to be secure in a specific game might not be secure in an arbitrary environment (composed with other protocols in parallel or in series).

Composable security on the other hand takes a different approach to phrasing the guarantees achieved by a protocol. Loosely speaking, a protocol is composable when it is shown to be secure in an arbitrarily adversarial environment\footnote{Of course, the environment may still be limited to ``efficient'' computations.}, and where secure means that it achieves a well-defined ideal (secure by definition) resource. This means the protocol retains the desired functionality even if it is composed of other instances of its own or a completely different protocol. There are several approaches which provide a general framework to study this cryptographic definitions~\cite{canetti2001universally,backes2003composable,maurer2011abstract}, but we will focus in this paper on Constructive Cryptography (CC) (also known under the term Abstract Cryptography (AC)). In this section, we provide relevant terminologies (mostly adapted to our protocol) required to analyse composable security in this framework, introduced by Maurer and Renner in~\cite{maurer2011abstract}. For more details, we refer readers to some of the previous works~\cite{maurer2011constructive,maurer2011abstract,dunjko2014composable,dunjko2016blind}.

The basic elements of AC are systems: objects with well-distinguished and labeled interfaces. The system uses interfaces to exchange information with the outside world and/or other systems. Systems are grouped in distinct classes: resources, converters, filters, and distinguisher.

\emph{Resource systems} (or $\cI$-resources) are devices with several interfaces in $\cI$, in general, each of them accessible by a single agent: each interface represents the actions that are accessible by that player. Resources are the central elements of CC, and they are used at the abstract level to specify the relevant properties of a protocol. Note that in this work we only consider resources with two interfaces $I = \{A,B\}$ because our protocol consists of two parties (one client $A$ and one server $B$).

A \emph{converter system}, on the other hand, is always limited to two interfaces, an inside and an outside one. Converters are usually attached to the interfaces of a resource (or a group of resource as already explained), and the name reflects the fact that a converter \emph{converts} the functionality of the resource's interface it is attached to into a new functionality on the outside. A resource having a converter attached to one of its interfaces continues to qualify as a resource, possibly equipped with new functionalities.  Usually, if $\alpha \in \Sigma$ is a converter and $\cR$ a resource, we write $\alpha^i\cR$ to denote new resource where the inner interface of $\alpha$ is connected to the interface $i$ of $\cR$, the outer interface of $\alpha$ being the new interface $i$. But because in our case we have two interfaces $\cI = \{A, B\}$, we will put the converter on the left of the resource when it is plugged on the interface $A$, and we will put the converter on the right of the resource when it is plugged on the interface $B$: $\alpha^A\cR$ is denoted $\alpha\cR$ while $\alpha^B\cR$ is $\cR\alpha$.
 
A \emph{filter} (usually denoted $\filter$) is a special converter used to force a honest behaviour on a given interface of a resource. They are usually used to prove the correctness of a protocol, as they describe what can be done in an honest run. They are removed when we want to provide full power to a cheating adversary or to a simulator. Usually, in order to keep the filter simple, the functionality accepts as a first message a bit $c$ which says if the party wants to behave honestly ($c=0$) or maliciously ($c=1$). That way, the filter $\filter^{c=0}$ (or simply $\filter$) just sends $c=0$ to the resource, and then forwards all the messages between it's inner and outer interface.

A \emph{distinguisher} helps to quantify the distance between resources. Given an $n$-interface resource $\cR$, a distinguisher $D \in \cD$ outputs a bit determined after interacting with the $n$ interfaces of $\cR$ (we denote by $D\cR$ this random variable). Then the distance (actually it is a pseudo-metric) between two resources $\cR$ and $\cS$ is defined by the best advantage $\eps$ a distinguisher $D \in \cD$ can achieve when trying to determine which resource it is interacting with. This leads to the following definition:
\begin{align}
  \cR \approx_\eps \cS :\Longleftrightarrow \Delta^{\cD}(\cR, \cS) \leq \eps
\end{align}
with $\Delta^{\cD}(\cR, \cS) = \sup_{D \in \cD} \Delta(D\cR, D\cS)$, where $\Delta(D\cR, D\cS)$ is the statistical distance between the distributions $D\cR$ and $D\cS$. Note that $\Delta^\cD$ defines a pseudo-metric: $\forall \eps > 0, (\cR, \cS, \cT) \in \Phi^3$,
\begin{align}
  \Delta^\cD(\cR, \cR) &= 0\\
  \Delta^\cD(\cR,\cS) &= \Delta^\cD(\cR,\cS)\\
  \Delta^\cD(\cR,\cS) &\leq \Delta^\cD(\cR,\cT) + \Delta^\cD(\cT,\cR)  
\end{align}

In the general Constructive Cryptography framework, we do not need to specify how the different systems are constructed, we just need to have some general properties on them: we basically require $\braket{\Phi, \Sigma}$ to be a \emph{cryptographic algebra}, and the pseudo-metric must be \emph{compatible} with $\braket{\Phi, \Sigma}$ (see for example  \cite[Sec. 4]{maurer2011constructive} for precise definitions). And as soon as $\cD$ can ``absorb'' the converters and the resources, then $\cD$ is compatible with $\braket{\Phi, \Sigma}$ \cite[Lem. 1]{maurer2011constructive}. That way, it is possible to derive different notions of security by just changing the sets $\Phi$ (resources), $\Sigma$ (converters) and $\cD$ (distinguisher) and making sure they respect these general properties. In the paper, we will focus mostly on two definitions (respecting the above properties): when all the systems (resources, converters, and distinguishers) are \emph{feasible} (in our case we mean they run in polynomial time on a quantum machine), denoted as $(\Phi^f, \Sigma^f, \cD^f)$ we will say that the security is \emph{computational}. If the systems are unbounded $(\Phi^u, \Sigma^u, \cD^u)$ we will refer to \emph{information-theoretic} security.

Note that the impossibility results presented in this paper apply in both computational and information-theoretic security, and because we only focus on these two settings, $\braket{\Phi, \Sigma}$ will always be a cryptographic algebra, and the pseudo-metric $\Delta^\cD$ is always compatible with it. When a property is valid only for one set of distinguishers, we will write this set above the $\approx$ sign, like for example $\cR \approx^{\cD^u}_\eps \cS$.

\noindent A main theorem is that any such construction achieve (general) composability:
\begin{lemma}[{\cite[Thm.~1]{maurer2011abstract}\cite[Thm.~3]{maurer2011constructive}}]
\sloppy The construction $\constructs{}$ is \emph{(generally) composable}, i.e. for all $(\eps, \eps') \in \R^+$, ${(\cR,\cS,\cT) \in \Phi^3}$, $\protocol \in \Sigma^2$:
  \begin{itemize}
  \item we have sequential composability: $(\cR \constructs{\protocol}{\eps} \cS \land \cS \constructs{\protocol'}{\eps'} \cT) \Rightarrow \cR \constructs{\protocol \circ \protocol'}{\eps + \eps'} \cT$,
  \item we have parallel composability: $(\cR \constructs{\protocol}{\eps} \cS \land \cR' \constructs{\protocol'}{\eps'} \cS') \Rightarrow \cR \| \cR' \constructs{\protocol | \protocol'}{\eps + \eps'} \cS \| \cS'$
  \item $\cR \constructs{\id}{0} \cR$
  \end{itemize}
  where $|$ (resp. $\circ$) represents the parallel (resp. serial) composition of protocols, $\|$ is the merging of resources, and $\id$ is the identity converter.
\end{lemma}

\section{QFactory: Remote State Preparation, Revisited\label{app:qfactory}}

The construction of the QFactory protocol relies on a family of functions with certain cryptographic properties, specifically, a 2-regular homomorphic-hardcore family of functions. For the formal definition of these properties, see \cite{cojocaru2019qfactory}.

We first begin by recalling the formal description of the protocol in \cref{subsec:QF4} and then in \cref{subsec:correctness_qfactory} and \cref{subsec:security_qfactory} we present the results concerning the correctness and security of QFactory.

\subsection{4-states and 8-states QFactory protocol}\label{subsec:QF4}

\begin{breakablealgorithm}
\caption{4-states QFactory: classical delegation of the BB84 states (\cite{cojocaru2019qfactory})} \label{protocol:qfactoryReal}
\vspace{\baselineskip}
\noindent \textbf{Requirements:} 
Public: A 2-regular homomorphic-hardcore family $\mathcal{F}$ with respect to $\{h_k\}$ and $d_0$. For simplicity, we will represent the sets $\mathcal{D}'$ (respectively $\mathcal{R}$) using $n$ (respectively $m$) bits strings: $\cD' = \{0, 1\}^{n}$, $\mathcal{R} = \{0, 1\}^{m}$. \\

\noindent\textbf{Stage 1: Preimages superposition} 
\begin{enumerate}
    \item Client runs the algorithm $(k,t_k) \leftarrow \text{Gen}_{\mathcal{F}}(1^n)$.
    \item Client instructs Server to prepare one register at $\otimes^n H\ket{0}$ and second register initiated at $\ket{0}^{m}$.
    \item Server receives k from the client and applies $U_{f_k}$ using the first register as control and the second as target.
    \item Server measures the second register in the computational basis, obtains the outcome $y$. The combined state is given by ${(\ket{x}+\ket{x'}) \otimes \ket{y}}$ with $f_k(x)=f_k(x')=y$ and $y\in \Ima f_k$. 
\end{enumerate}

\noindent\textbf{Stage 2: Output preparation} 
\begin{enumerate}
    \item Server applies $U_{h_k}$ on the preimage register $\ket{x}+\ket{x'}$ as control and another qubit initiated at $\ket{0}$ as target. Then, measures all the qubits, but the target in the $\{\frac{1}{\sqrt{2}}(\Ket{0} \pm \Ket{1})\}$ basis, obtaining the outcome $b = (b_1, ..., b_{n})$. Now, the Server returns both $y$ and $b$ to the Client. 
    \item Client using the trapdoor $t_k$ computes the preimages of $y$:
    \begin{itemize}
    \item if $y$ does not have exactly two preimages $x,x'$ (the server is cheating with overwhelming probability), defines $B_1 = d_0(t_k)$, and chooses $B_2 \in \{0,1\}$ uniformly at random
    \item if $y$ has exactly two preimages $x,x'$, defines $B_1 = h_k(x) \xor h_k(x') = d_0(t_k)$, and $B_2$. 
    \end{itemize}
\end{enumerate}

\noindent \textbf{Output:} The quantum state that the Server has generated is (with overwhelming probability~\footnote{\label{note}As for the previous protocol, the probability comes from the probability of $\mathcal{F}$ being a 2-regular homomorphic-hardcore family of functions}) the BB84 state $\ket{\quout} = H^{B_1} X^{B_2} \ket{0}$  (see \cref{eq:expressions_exp1} and \cref{eq:expressions_exp2} for the exact value of $B_1$ and $B_2$). The output of the Server is a quantum state $\ket{\quout}$ and the output of the Client is given by $(B_1,B_2)$ ($2$ bits).
\vspace{\baselineskip}
\end{breakablealgorithm}

\vspace{\baselineskip}

\begin{breakablealgorithm}
\caption{8-states QFactory: classical delegation of the $\ket{+_{\theta}}$ states (\textup{\cite{cojocaru2019qfactory}})} \label{protocol:QFactory4to8}
\vspace{\baselineskip}
\noindent\textbf{Requirements:} Same as in Protocol~\ref{protocol:qfactoryReal}\\

\noindent\textbf{Input:} Client runs twice the algorithm ${Gen}_{\mathcal{F}}(1^n)$, obtaining $(k^1, t^1_k), (k^2, t^2_k)$. Client keeps $t_k^1, t_k^2$ private.\\ 

\noindent\textbf{Protocol Steps:}
\begin{enumerate}
\item Client runs 4-states QFactory Protocol~\ref{protocol:qfactoryReal} to obtain a state $\ket{\quin_1}$ and a "rotated" 4-states QFactory to obtain a state $\ket{\quin_2}$ (by rotated 4-states QFactory we mean a 4-states QFactory, but where the last set of measurements in the $\Ket{\pm}$ basis is replaced by measurements in the $\Ket{\pm_{ \frac{\pi}{2}} }$ basis). 
\item Client records measurement outcomes $(y^1,b^1)$, $(y^2, b^2)$ and computes and stores the corresponding  indices of the output states of the 2 runs of 4-states QFactory protocol: $(B_1, B_2)$ for $\ket{\quin_1}$ and $(B_1', B_2')$ for $\ket{\quin_2}$.  
\item Client instructs Server to apply the Merge Gadget in Fig. \ref{fig:gadget} (\cite{cojocaru2019qfactory}) on the states $\ket{\quin_1}$, $\ket{\quin_2}$. 
\item Server returns the 2 measurement results $s_1$, $s_2$. 
\item Client using $(B_1, B_2)$, $(B_1', B_2')$, $s_1$, $s_2$ computes the index $L = L_1L_2L_3 \in \{0, 1\}^3$ of the output state (see \cref{eq:expressions_l_1}, \cref{eq:expressions_l_2}, and \cref{eq:expressions_l_3} for the exact value of $L_1$, $L_2$, and $L_3$, respectively.)
\end{enumerate}

\noindent\textbf{Output:} The output of the Server is (with overwhelming probability) a quantum state $\ket{\quout} :=  \ket{+_{L\frac{\pi}{4}}}$ and the output of the Client is given by $L$ ($3$ bits). 
\vspace{\baselineskip}
\end{breakablealgorithm}

\begin{figure}[ht]
\centering
\includegraphics[trim=5cm 21cm 4cm 4cm, clip=true]{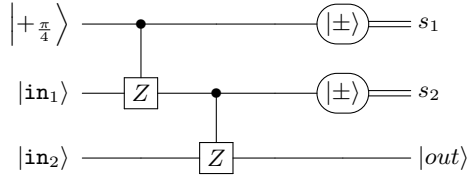}
\caption{Merge Gadget (Taken from \textup{\cite{cojocaru2019qfactory}})}
\label{fig:gadget}
\end{figure}

\subsection{Correctness of QFactory \label{subsec:correctness_qfactory}}

In an honest run, the description of the output state of the protocol depends on measurement results $y \in \Ima f_k$ and $b$, but also on the 2 preimages $x$ and $x'$ of $y$.

The output state of 4-states QFactory belongs to the set of states $\{\Ket{0}, \Ket{1}, \Ket{+}, \Ket{-}\}$ and its exact description is the following:

\begin{theorem}[4-states QFactory is correct (\textup{\cite{cojocaru2019qfactory}})]\label{thm:correctness}
  In an honest run, with overwhelming probability the output state $\ket{\quout}$ of the 4-states QFactory Protocol~\ref{protocol:qfactoryReal} is a BB84 state whose basis is $B_1 = h_k(x) \xor h_k(x') = d_0$, and:
\begin{itemize}
\item if $d_0 = 0$, then the state is $\ket{h_k(x)}$ (computational basis, also equal to $\ket{h_k(x')}$)
\item if $d_0 = 1$, then if $\sum_i b_i \cdot (x_i \xor x'_i) = 0 \bmod 2$, the state is $\ket{+}$, otherwise the state is $\ket{-}$ (Hadamard basis).
\end{itemize}
i.e.
\begin{align}
  \ket{\quout} &= H^{B_1} X^{B_2} \ket{0} \label{eq:correct_output}
\end{align}
with
\begin{align}
  B_1 &= h_k(x) \xor h_k(x') = d_0 \label{eq:expressions_exp1} \\
  B_2 &= (d_0 \times (b \cdot (x \xor x'))) \xor h(x)h(x') \label{eq:expressions_exp2}
\end{align}
(the inner product is taken modulo 2, and $x \xor x'$ is a bitwise xor)
\end{theorem}

\begin{theorem}[8-states QFactory is correct (\textup{\cite{cojocaru2019qfactory}})]\label{thm:8_statescorrectness}
In an honest run, the Output state of the 8-states QFactory Protocol is of the form $\Ket{+_{L \cdot \frac{\pi}{4}}}$, where $L = L_1L_2L_3 \in \{0, 1\}^3$, defined as:
\begin{align}
L_1 &= B_2' \oplus B_2 \oplus [B_1 \cdot (s_1 \oplus s_2)]  \label{eq:expressions_l_1} \\
L_2 &= B_1' \oplus [(B_2 \oplus s_2) \cdot B_1] \label{eq:expressions_l_2}\\
L_3 &=  B_1
\label{eq:expressions_l_3}
\end{align}
\end{theorem}

\subsection{Security of QFactory \label{subsec:security_qfactory}}

In any run of the protocol, honest or malicious, the state that the client believes that the server has is given by \cref{thm:correctness}. Therefore, the task that a malicious server wants to achieve, is to be able to guess, as good as he can, the description of the output state that the client (based on the public communication) thinks the server has produced. In particular, in our case, the server needs to guess the bit $B_1$ (corresponding to the basis) of the (honest) output state.

\begin{definition}[4 states basis blindness]\label{def:4basisblind}
  We say that a protocol $(\pi_A, \pi_B)$ achieves \textbf{basis-blindness} with respect to an ideal list of 4 states \\
  $S = \{S_{B_1,B_2}\}_{(B_1,B_2) \in \{0,1\}^2}$ if:
  \begin{itemize}
  \item $S$ is the set of states that the protocol outputs, i.e.:
    \[\pr{\ket{\phi} = S_{B_1B_2} \in S \mid ((B_1,B_2),\ket{\phi}) \leftarrow (\pi_A \| \pi_B) } \geq 1 - \negl\]
  \item and no information is leaked about the index bit $B_1$ of the output state of the protocol, i.e for all QPT adversary $\cA$:
  \[\pr{ B_1 = \tilde{B_1} \mid ((B_1, B_2), \tilde{B_1}) \leftarrow (\pi_A \| \cA)} \leq 1/2 + \negl \]
  \end{itemize}
\end{definition}

\begin{theorem}[4-states QFactory is secure (\textup{\cite{cojocaru2019qfactory}})] \label{thm:security_qfac_2_0}
  Protocol~\ref{protocol:qfactoryReal} satisfies $4$-states basis blindness with respect to the ideal list of states \\
  $S = \{H^{B_1}X^{B_2}\ket{0}\}_{B_1,B_2} = \{\Ket{0}, \Ket{1}, \Ket{+}, \Ket{-}\}$.
\end{theorem}

\begin{definition}[8 states basis blindness]\label{def:8basisblind}
\sloppy  Similarly, we say that a protocol $(\pi_A, \pi_B)$ achieves \textbf{basis-blindness} with respect to an ideal list of 8 states $S = \{S_{L_1, L_2, L_3}\}_{(L_1, L_2, L_3) \in \{0,1\}^3}$ if:
  \begin{itemize}
  \item $S$ is the set of states that the protocol outputs, i.e.:
    \[\pr{\ket{\phi} = S_{L_1,L_2,L_3} \in S \mid ((L_1,L_2,L_3),\ket{\phi}) \leftarrow (\pi_A \| \pi_B) } =1\]
  \item and if no information is leaked about the ``basis'' bits $(L_2,L_3)$ of the output state of the protocol, i.e for all QPT adversary $\cA$:
    \[\pr{ L_2 = \tilde{L_2} \text{ and } L_3 = \tilde{L_3} \mid ((L_1,L_2,L_3), (\tilde{L_2},\tilde{L_3})) \leftarrow (\pi_A \| \cA)} \leq 1/4 + \negl \]
  \end{itemize}
\end{definition}

\begin{theorem}[8-states QFactory is secure (\textup{\cite{cojocaru2019qfactory}})]\label{thm:security_qfac_2_1}
\sloppy Protocol~\ref{protocol:QFactory4to8} satisfies $8$-state basis blindness with respect to the ideal set of states ${S = \{\ket{+_{\pi L /4}}\}_{L \in \{0,\dots,7\}} = \{\ket{+}, \ket{+_{\frac{\pi}{4}}}, .., \ket{+_{\frac{7\pi}{4}}}\}}$.
\end{theorem}

\section{Distance Measures for Quantum States}

\begin{lemma}\label{lemma:trace_formula}
	For any two self-adjoint trace-class operators $\rho, \sigma$ it holds that
	\begin{align*}
		\Tr(\rho\sigma) = \frac{1}{2} \left[ \Tr(\rho^2) + \Tr(\sigma^2) \right] - \frac{1}{2} \left\| \rho - \sigma \right\|_{\text{HS}}^2,
	\end{align*}
	where the Hilbert-Schmidt norm is defined as
	\begin{align*}
		\| A \|_{\text{HS}} = \sqrt{\Tr(A^\ast A)}.
	\end{align*}
\end{lemma}
\begin{proof}
	This follows directly from the relation
	\begin{align*}
		(\rho - \sigma)^2 = \rho^2 - \rho\sigma - \sigma\rho + \sigma^2
	\end{align*}
	and the fact that $\rho$ and $\sigma$ are self-adjoint operators.
\end{proof}

The following lemma formalizes the following statement: If $\Tr(\rho\sigma)$ is close to $1$, then both $\rho$ and $\sigma$ must be almost pure, and $\rho$ and $\sigma$ must be close. Note that \cref{lemma:trace_properties} holds in particular for density matrices $\rho$ and $\sigma$, despite being stated for a more general class of operators.

\begin{lemma}\label{lemma:trace_properties}
	Let $\varepsilon \geq 0$ and $\Tr \left( \rho\sigma \right) \geq 1 - \varepsilon$ for two self-adjoint, positive semi-definite operators $\rho, \sigma$ with trace less than 1. Then, it holds that
	\begin{enumerate}
		\item $\Tr \left( \rho^2 \right) \geq 1 - 2\varepsilon$,
		\item $\Tr \left( \sigma^2 \right) \geq 1 - 2\varepsilon$, and
		\item $\left\| \rho - \sigma \right\|_{\text{HS}} \leq \sqrt{2\varepsilon}$.
	\end{enumerate}
\end{lemma}
\begin{proof}
	\begin{enumerate}
		\item With the formula from \cref{lemma:trace_formula}, we infer that
			\begin{align*}
				\Tr(\rho\sigma) \leq \frac{1}{2} \left[ \Tr(\rho^2) + \Tr(\sigma^2) \right] \leq \frac{1}{2} \left[ \Tr(\rho^2) + 1 \right],
			\end{align*}
			using the non-negativity of the Hilbert-Schmidt norm and the fact that $\Tr \left( \sigma^2 \right) \leq 1$. Hence,
			\begin{align*}
				\Tr \left( \rho^2 \right) \geq 2 \Tr \left( \rho\sigma \right) - 1 \geq 1 - 2\varepsilon .
			\end{align*}
		\item Analogously to 1.
		\item Using $\Tr \left( \rho^2 \right) \leq 1$ and $\Tr \left( \sigma^2 \right) \leq 1$, we obtain
			\begin{align*}
				\Tr \left( \rho\sigma \right) &\leq 1 - \frac{1}{2} \left\| \rho - \sigma \right\|_{\text{HS}}^2 \\
				\Rightarrow \left\| \rho - \sigma \right\|_{\text{HS}}^2 &\leq 2 \left( 1 - \Tr \left( \rho\sigma \right) \right) \leq 2 \varepsilon,
			\end{align*}
		    which implies the claim.
	\end{enumerate}
\end{proof}

\begin{lemma}\label{lemma:equivalence_purity_closeness}
  Let $\lambda$ be a security parameter and let $\rho, \sigma$ be two density matrices of finite and fixed dimension. Then, the following statements are equivalent:
  \begin{enumerate}
    \item $\Tr \left( \rho^2 \right) \geq 1 - \operatorname{negl}(\lambda)$, $\Tr \left( \sigma^2 \right) \geq 1 - \operatorname{negl}(\lambda)$, and $\operatorname{TD}\left( \rho - \sigma \right) \leq \operatorname{negl}(\lambda)$,
    \item $\Tr \left( \rho \sigma \right) \geq 1 - \operatorname{negl}(\lambda)$,
  \end{enumerate}
  where $\operatorname{TD}$ denotes the trace distance.
\end{lemma}
\begin{proof}
    One direction of the equivalence follows directly from \cref{lemma:trace_properties}. The other direction follows from the formula in \cref{lemma:trace_formula} and the fact that in finite-dimensional spaces the trace norm is equivalent to the Hilbert-Schmidt norm.
\end{proof}

\begin{lemma}\label{lemma:trace_transitivity}
	Let $\varepsilon_1, \varepsilon_2 \geq 0$. Let further $\Tr \left( \rho_1 \rho_2 \right) \geq 1 - \varepsilon_1$ and $\Tr \left( \rho_2 \rho_3 \right) \geq 1 - \varepsilon_2$ for self-adjoint, positive semi-definite operators $\rho_1, \rho_2, \rho_3$ with trace less than 1. Then it holds that $\Tr \left( \rho_1 \rho_3 \right) \geq 1 - 3 \left( \varepsilon_1 + \varepsilon_2 \right)$.
\end{lemma}
\begin{proof}
	From \cref{lemma:trace_properties} we know that $\Tr \left( \rho_1^2 \right) \geq 1 - 2\varepsilon_1$, $\Tr \left( \rho_3^2 \right) \geq 1 - 2\varepsilon_2$, and
	\begin{align*}
		\left\| \rho_1 - \rho_2 \right\|_{\text{HS}} \leq \sqrt{2\varepsilon_1}, \quad \left\| \rho_2 - \rho_3 \right\|_{\text{HS}} \leq \sqrt{2\varepsilon_2}.
	\end{align*}
	By the triangle inequality for the Hilbert-Schmidt norm, it follows readily that
	\begin{align*}
		\left\| \rho_1 - \rho_3 \right\|_{\text{HS}} \leq \sqrt{2\varepsilon_1} + \sqrt{2\varepsilon_2}
	\end{align*}
	and therefore
	\begin{align*}
		\left\| \rho_1 - \rho_3 \right\|_{\text{HS}}^2 &\leq \left( \sqrt{2\varepsilon_1} + \sqrt{2\varepsilon_2} \right)^2 = 2 \varepsilon_1 + 2 \varepsilon_2 + 4 \sqrt{\varepsilon_1} \sqrt{\varepsilon_2} \leq 4 \left( \varepsilon_1 + \varepsilon_2 \right)
	\end{align*}
	where we applied the inequality of the geometric mean to obtain the last bound. Using the formula from \cref{lemma:trace_formula}, we then conclude that
	\begin{align*}
		\Tr \left( \rho_1 \rho_3 \right) &= \frac{1}{2} \left[ \Tr(\rho_1^2) + \Tr(\rho_3^2) \right] - \frac{1}{2} \left\| \rho_1 - \rho_3 \right\|_{\text{HS}}^2 \\
		&\geq \frac{1}{2} \left[ 1 - 2\varepsilon_1 + 1 - 2\varepsilon_2 \right] - \frac{1}{2} 4 \left( \varepsilon_1 + \varepsilon_2 \right) \geq 1 - 3 \left( \varepsilon_1 + \varepsilon_2 \right).
	\end{align*}
\end{proof}

\end{document}